\documentclass[11pt]{article}

\usepackage{graphicx}
\usepackage{caption}
\usepackage{subcaption}
\usepackage{xspace}
\usepackage{url}
\usepackage{mathtools}
\usepackage{amssymb}
\usepackage{amsthm}
\usepackage{empheq}
\usepackage{latexsym}
\usepackage{enumitem}
\usepackage{eurosym}
\usepackage{dsfont}
\usepackage{appendix}
\usepackage{color} 
\usepackage[unicode]{hyperref}
\usepackage{frcursive}
\usepackage[utf8]{inputenc}
\usepackage[T1]{fontenc}
\usepackage{geometry}
\usepackage{multirow}
\usepackage{todonotes}
\usepackage{lmodern}
\usepackage{anyfontsize}
\usepackage{stmaryrd}
\usepackage{natbib}
\usepackage{cleveref}
\usepackage[english]{babel}
\usepackage[english=british]{csquotes}
\usepackage{mathtools}
\usepackage{accents}
\usepackage{color}

\setcitestyle{numbers,open={[},close={]}}

\definecolor{red}{rgb}{0.7,0.15,0.15}
\definecolor{green}{rgb}{0,0.5,0}
\definecolor{blue}{rgb}{0,0,0.7}
\hypersetup{colorlinks, linkcolor={black}, citecolor={black}, urlcolor={black}}

\makeatletter \@addtoreset{equation}{section}

\newtheorem{theorem}{Theorem}[section]
\newtheorem{assumption}[theorem]{Assumption}

\newtheorem{lemma}[theorem]{Lemma}
\newtheorem{proposition}[theorem]{Proposition}

\newtheorem{definition}[theorem]{Definition}
\newtheorem{remark}[theorem]{Remark}

\DeclareUnicodeCharacter{014D}{\=o}
\newcommand{\diff}{\,\mathrm{d}}

\newcommand\cF{\mathcal F}

\newcommand\cW{\mathcal W}

\newcommand\PP{\mathbb P}

\newcommand\RR{\mathbb R}

\providecommand{\Keywords}[1]{\textbf{Keywords} #1}
\providecommand{\abstract}[1]{\textbf{Abstract} #1}
\providecommand{\msc}[1]{\textbf{Mathematics Subject Classification} #1}
\providecommand{\statement}[1]{\textbf{Statements and Declarations} #1}

\def \E{\mathbb{E}}

\def \L{\mathbb{L}}

\def \N{\mathbb{N}}

\def \P{\mathbb{P}}

\def \R{\mathbb{R}}

\def \Z{\mathbb{Z}}

\def\Bc{\mathcal{B}}

\def\d{\mathrm{d}}

\DeclareMathOperator*{\argmin}{arg\,min}

\DeclareMathOperator*{\esssup}{ess\,sup}

\begin{document}
\title{Non-asymptotic estimation of risk measures using stochastic gradient Langevin dynamics}
\author{Jiarui Chu\footnote{Corresponding Author: Jiarui Chu\newline \hspace*{1.5em} Affiliation: Princeton University, ORFE\newline  \hspace*{1.8em}E-mail Address:jiaruic@princeton.edu} $\qquad \qquad$  Ludovic Tangpi\footnote{Princeton University, ORFE, ludovic.tangpi@princeton.edu}}

\maketitle

\abstract{
In this paper we will study the approximation of some law invariant risk measures. As a starting point, we approximate the average value at risk using stochastic gradient Langevin dynamics, which can be seen as a variant of the stochastic gradient descent algorithm. Further, the Kusuoka's spectral representation allows us to bootstrap the estimation of the average value at risk to extend the algorithm to general law invariant risk measures. We will present both theoretical, non-asymptotic convergence rates of the approximation algorithm and numerical simulations.

}

\msc{91G70 $\cdot$ 90C90}

\statement{The authors gratefully acknowledge support from the NSF grant DMS-2005832. The authors have no competing interests to declare. }
\setlength{\parindent}{0pt}
\section{Introduction}\label{sec:intro}
Every financial decision involves some degree of risk. Quantifying risk associated with a future random outcome allows organizations to compare financial decisions and develop risk management plans to prepare for potential loss and uncertainty. By the seminal work of \citeauthor*{Artzner1999} \cite{Artzner1999}, the canonical way to quantify the riskiness of a \emph{random} financial position $X$ is to compute the \emph{number} $\rho(X)$ for a convex risk measure $\rho$, whose definition we recall.
\begin{definition}
	A mapping $\rho: \L^{\infty} \to \R$ is a convex risk measure if it satisfies the following conditions for all $X, Y \in \mathcal{X}$:
	\begin{itemize}
    	\item translation invariance: $\rho(X+m)=\rho(X)-m$ for all $m \in \R$
    	\item monotonicity: $\rho(X) \geq \rho(Y)$ if $X \leq Y$
    	\item convexity: $\rho(\lambda X+(1-\lambda) Y)\leq \lambda \rho(X)+(1-\lambda)\rho(Y)$ for $\lambda \in [0,1]$.
	\end{itemize}\label{risk}
\end{definition}
Intuitively\footnote{In the rest of the paper, for notation simplicity, we will assume risk measures to be increasing.
That is, we work with $\rho(-X)$.
This does not restrict the generality.}, $\rho(X)$ measures the minimum amount of capital that should be added to the current financial position $X$ to make it acceptable. 
Due to its fundamental importance in quantitative finance, the theory of risk measures (sometimes called quantitative risk management) has been extensively developed.
We refer for instance to \cite{Del12,foellmer02,Jou-Sch-Tou,Cheridito2006,Kupper2009,cheridito2009risk,fritelli03,kusuoka2001law} for a few milestones and the influential textbooks of \citeauthor*{McNeil15} \cite{McNeil15} and \citeauthor*{foellmer01} \cite{foellmer01} for overviews.

\medskip

An important problem for risk managers in practice is to efficiently simulate the number $\rho(X)$ for a financial position $X$ and a risk measure $\rho$.
The difficulty here stems from the fact that, unless $\rho$ is a ``simple enough'' risk measure and the law of $X$ belongs to a tractable family of distributions, there are no closed form formula allowing to compute $\rho(X)$.
The goal of this paper is to develop a method allowing to \emph{numerically simulate} the riskiness $\rho(X)$ for general convex risk measures, and when the law of $X$ is not necessarily known (as it is the case in practical applications).

\medskip

One commonly used measure of the riskiness of a financial position is the value at risk (VaR). 
For a given risk intolerance $u \in (0,1)$, the value at risk $\mathrm{VaR}_u(X)$ of $X$ is the $(1-u)$-quantile of the distribution of $X$. 
Despite the various shortcomings of this measure of risk documented by the academic community \cite{McNeil15}, $\mathrm{VaR}$ remains the standard in the banking industry, and due to its widespread use, the computation of $\mathrm{VaR}$ has been extensively studied. 
We refer interested readers for instance to \citep{glasserman,hoogerheide,butler,fan}, and references therein for various simulation techniques. 
Recommendations \cite{Basel3} from the Basel Committee on Banking Supervision which advises on risk management for financial institutions have revived the development of convex risk measures such as the average value at risk (AVaR), also called conditional value at risk or expected shortfall.
This risk measure is the expected loss given that losses are greater than or equal to the $\mathrm{VaR}$. 
That is, $\mathrm{AVaR}_u$ is given by:
\begin{equation}\label{definition}
	\mathrm{AVaR}_u(X) :=  \E [X|X> \mathrm{VaR}_u(X)].
\end{equation}
For general distributions of $X$, $\mathrm{AVaR}_u(X)$ usually does not have closed form expressions. 
Therefore, in practice, numerical estimations are often required.
As a result, the estimation of $\mathrm{AVaR}$ has received considerable attention. 
We refer for instance to works by \citeauthor*{Eck-Kup19} \cite{Eck-Kup19} and \citeauthor*{Bueh-Gon-Tei-Wood18} \cite{Bueh-Gon-Tei-Wood18} in which (among other things) the simulation of optimized certainty equivalents (of which $\mathrm{AVaR}$ is a particular case) are considered using deep learning techniques.
One approximation technique for the $\mathrm{AVaR}$ is based on Monte-Carlo type algorithm.
In this direction, let us refer for instance to \citeauthor*{avar} \cite{avar} and \citeauthor*{Chen08} \cite{Chen08} on Monte Carlo estimation of $\mathrm{VaR}$ and $\mathrm{AVaR}$, \citeauthor*{Zhou} \cite{Zhou} on nested Monte Carlo estimation. 
More recently, motivated by developments of gradient descent methods in stochastic optimization, (in particular the stochastic Langevin gradient descent (SGLD) technique), \citeauthor*{sgld} \cite{sgld} provide non-asymptotic error bounds for the estimation of $\mathrm{AVaR}$.
Other works developing such gradient descent techniques in the context of risk management include \citeauthor*{iyengar} \cite{iyengar},
\citeauthor*{tamar} \cite{tamar} and \citeauthor*{soma} \cite{soma}. 
Essentially, these papers take advantage of new developments in machine learning and optimization, see e.g.\ \citeauthor*{allen} \cite{allen}, \citeauthor*{gelfand} \cite{gelfand}, 
\citeauthor*{nesterov} \cite{nesterov} and \citeauthor*{Raginsky} \cite{Raginsky}. Let us also mention the recent work of \citeauthor*{Reppen-Soner21} \cite{Reppen-Soner21} who develop a data--driven approach based on ideas from learning theory.

\medskip

In this work we go beyond the numerical simulation of $\mathrm{AVaR}$ by extending stochastic gradient descent type techniques to compute a large family of risk measures, including the $\mathrm{AVaR}$. 
We are interested in this work in deriving explicit (non--asymptotic) error estimates for the approximation.
We will restrict our attention to law-invariant convex risk measures (whose definition we recall below), since in practice, only the law of a financial position can be (approximately) observed. 
In fact, the requirement for a risk measure to be law-invariant is natural and is satisfied by most risk measures\footnote{All risk measures considered in this work will be implicitly assumed to be convex law--invariant risk measures.}.
\begin{definition}\citep{law}
	A risk measure $\rho$ is law-invariant if for all $X,Y$ with the same distribution, we have $\rho(X)=\rho(Y).$\label{law}
\end{definition}
To the best of our knowledge, the papers considering (non-parametric) estimation of general convex risk measures are \citeauthor*{Weber07} \cite{Weber07}, \citeauthor*{belomestny2012central} \cite{belomestny2012central} and \citeauthor*{Ludo} \cite{Ludo}.
These papers consider a (data-driven) Monte-Carlo estimation method by proposing a plug-in estimator based on the empirical measure of the historical observations of the underlying distribution of the random outcome.
\citeauthor*{Weber07} \cite{Weber07} proves a large deviation theorem, and \citeauthor*{belomestny2012central} \cite{belomestny2012central} provide a central limit theorem.
Note that both of these papers give asymptotic estimation results.
Bartl and Tangpi \cite{Ludo} provide sharp non-asymptotic convergence rates for the estimation.

\medskip
To estimate general law-invariant convex risk measures, we rely on the \citeauthor*{kusuoka2001law}'s spectral representation \cite{kusuoka2001law}. 
Intuitively, this representation says that any law invariant risk measure can be constructed as an integral of the $\mathrm{AVaR}$ risk measure. Therefore, the first step of our approximation of general law--invariant risk measures is to estimate $\mathrm{AVaR}$. 
Since we would like to analyze approximation algorithms for the risk of claims with possibly non-convex payoffs, we employ the idea of \citeauthor*{Raginsky} \cite{Raginsky} and use the stochastic gradient Langevin dynamic which, essentially, adds a Gaussian noise to the unbiased estimate of the gradient in stochastic gradient descent. 
To quantify the distance between the estimator and the true value of the risk measure, we present non--asymptotic rates on the mean squared estimation error both in the case of a $\mathrm{AVaR}$, and of general law--invariant risk measures. 
The proof of the mean squared error of estimating the $\mathrm{AVaR}$ makes use of the observation that the SGLD algorithm is a variant of the Euler-Maruyama discretization of the solution of the Langevin stochastic differential equation (SDE). 
This observation allows us to use results on the convergence rate of the Euler-Maruyama scheme and classical techniques of deriving the convergence rate of the solution of the Langevin SDE to the invariant measure. For the rate on the mean squared estimation error of the general case, our proof relies heavily on \citeauthor*{kusuoka2001law}'s representation which allows to build general law--invariant risk measures from $\mathrm{AVaR}$. 

\medskip

Beyond our theoretical guarantees for the convergence of approximation algorithms for general convex risk measures, the present work also contributes to the non-convex optimization literature in that we propose a new proof for the convergence of SGLD algorithms for some non--convex objective functions.
The idea is essentially to reduce the problem into the analysis of contractivity properties of the semi--group originating from a Langevin diffusion with non--convex potential.
This problem was notably investigated by \citeauthor*{EberlePTRF} \cite{EberlePTRF}.

\medskip
The paper is organized as follows:
We start by describing the approximation techniques and presenting the main results in \Cref{eq:approx.meth}. In the same section, we also present numerical results on the estimation of AVaR. In \Cref{sec:proofs.avar}, we prove the rates on the mean squared error for the estimation of AVaR. The derivation of the mean squared error for the estimation of a general law-invariant risk measure is done in \Cref{sec:proofs.genRM}. 

\medskip
\textbf{Notations:} Let $\mathbb{N}^\star :=\mathbb{N}\setminus\{0\}$ and let $\mathbb{R}_+^\star $ be the set of real positive numbers. Fix an arbitrary Polish space $E$ endowed with a metric $d_E$. Throughout this paper, for every $p$-dimensional $E$-valued vector $e$ with $p\in \mathbb{N}^\star $, we denote by $e^{1},\ldots,e^{p}$ its coordinates.
For $(\alpha,\beta) \in \R^p\times\R^p$, we also denote by $\alpha\cdot \beta$ the usual inner product, with associated norm $\|\cdot\|$, which we simplify to $|\cdot|$ when $p$ is equal to $1$. 
For any $(\ell,c)\in\mathbb N^\star \times\mathbb N^\star $, $E^{\ell\times c}$ will denote the space of $\ell\times c$ matrices with $E$-valued entries.

\medskip
Let $\Bc(E)$ be the Borel $\sigma$-algebra on $E$ (for the topology generated by the metric $d_E$ on $E)$. For any $p\geq 1$, for any two probability measures  $\mu$ and $\nu$ on $(E,\Bc(E))$ with finite $p$-moments, we denote by $\cW_{p}(\mu, \nu)$ the $p$-Wasserstein distance between $\mu$ and $\nu$, that is
\begin{equation*}
	\cW_{p}(\mu, \nu) := \bigg(\inf_{\alpha\in\Gamma(\mu,\nu)}\int_{E\times E}d(x,y)^p\alpha(\mathrm{d}x,\mathrm{d}y) \bigg)^{1/p},
\end{equation*}
where the infimum is taken over the set $\Gamma(\mu,\nu)$ of all couplings $\pi$ of $\mu$ and $\nu$, that is, probability measures on $\big(E^2,\Bc(E)^{\otimes 2}\big)$ with marginals $\mu$ and $\nu$ on the first and second factors respectively.

\section{Approximation technique and main results}
 \label{eq:approx.meth}
In this section we rigorously describe the approximation method developed in this article as well as our main results.
Throughout, we fix a probability space $(\Omega, \cF, \P)$ on which all random variables will be defined, unless otherwise stated.
Let us denote by $\L^\infty$ the space of essentially bounded random variables on this probability space.
The starting point of our method is based on the following spectral representation of law-invariant risk measures due to \citeauthor*{kusuoka2001law} \cite{kusuoka2001law}:
\begin{theorem}
	A mapping $\rho:\L^\infty \to \R$ is a law-invariant risk measure if and only if it satisfies
	\begin{equation}\label{Kusuoka}
		\rho(X)=\sup_{\gamma\in \mathcal{M}}\bigg(\int_{[0,1)}\text{AVaR}_u(X)\gamma(du)-\beta(\gamma)\bigg), \text{ for all } X \in \L^{\infty}
	\end{equation}
	for some functional $\beta:\mathcal{M} \to [0,\infty)$, where $\mathcal{M}$ is the set of all Borel probability measures on $[0,1]$.\label{thm:Kusuoka}
\end{theorem}
In fact, this spectral representation suggests that the risk measure $\mathrm{AVaR}$ is the ``basic building block'' allowing to construct all convex law--invariant risk measures.
Thus, the idea will be to propose an approximation algorithm for $\mathrm{AVaR}$ that will be later bootstrapped to derive an algorithm for general law invariant risk measures.
This approach is also used in \citep{Ludo} for a very different approximation method.

\subsection{Approximation of average value at risk}
Let us first focus on estimating $\mathrm{AVaR}$.
For this purpose, recall (see e.g.\ \citep[Proposition 4.51]{foellmer01}) that for every $X \in \L^\infty$ and $u \in (0,1)$,  $\mathrm{AVaR}_u(X)$ takes the form 
\begin{equation*}
	\text{AVaR}_u(X)=\inf_{q\in \R}\Big(\frac{1}{1-u}\E[({}X-q)^+]+q\Big).    
\end{equation*}
In other words, $\mathrm{AVaR}_u(X)$ is nothing but the value of a stochastic optimization problem.
In most financial applications, the contingent claim $X$ whose risk is assessed is of the form $X = f(r,S)$ where $S = (S^1,\dots,  S^d)$ is a $d$--dimensional random vector of risk factors, and $r \in \mathcal{X}$, see e.g.\ \citep[Section 2.1]{McNeil15} for details. Note that here the space $\mathcal{X}$ can be infinite dimensional. A standard practice to approach the infinite dimensional case is to use neural networks for approximation, which leads to non-convex objective functions. Therefore, we allow $f$ to be non-convex with certain regularity conditions. 
A standard example arises when $X$ is the profit and loss (P$\&$L) of an investment strategy.
In this case, $r$ is the portfolio and the random vector $S$ represents (increments) of the stock prices. That is, $f(r,S):= \sum_{i=1}^dr_i(S_1^i-S^i_0)$ where $S^i_1$ and $S^i_0$ are the values of the stock $i$ at times $1$ and $0$, respectively.
Hence, we let $r$ be in a compact and convex set $A \subseteq \R^{d-1}$, and our goal will be to estimate the value of the (multi-dimensional) risk minimization problem
\begin{align}\notag
	\overline{\mathrm{AVaR}_u}(f) &:=\inf_{r \in A}\mathrm{AVaR}_u(f(r, S))\\
	\label{opt}
	&=\inf_{r \in A, q\in \R}\Big(\frac{1}{1-u}\E\Big[\big(f(r, S)-q\big)^+\Big]+q\Big).
\end{align}
A natural way to numerically solve such problems is by gradient descent. However, when the dataset is large, gradient descent usually does not perform well, since computing the gradient on the full dataset at each iteration is computationally expensive. 

Among others, one method that has been proposed to get around the high computational cost of gradient descent is the stochastic gradient descent (SGD) algorithm, which replaces the true gradient with an unbiased estimate calculated from a random subset of the data. A more  recent approach, called the Stochastic Langevin Gradient Descent, injects a random noise to an unbiased estimate of the gradient at each iteration of the SGD algorithm. Originally introduced by \citeauthor*{welling} \cite{welling} as a tool for Bayesian posterior sampling on large scale and high dimensional datasets, SGLD maintains the scalability property of SGD, and has a few advantages over the SGD: By adding a noise to SGD, SGLD navigates out of saddle points and local minima more easily \citep{bhardwaj}, outperforms SGD in terms of accuracy \citep{Deng}, and overcomes the curse of dimensionality \citep{accuracy}. 
Moreover, SGLD also applies to cases where the objective function is non-convex but sufficiently regular \cite{gelfand} \cite{Raginsky}. 

\medskip

We will apply the SGLD in the present context of estimation of $\mathrm{AVaR}$.
Recall that our goal is to solve the optimization problem given in equation (\ref{opt}). Let $z:=(r,q)$, and consider the (objective) function
$$
	\widetilde{L}(r,q) := \frac{1}{1-u}\E\Big[\big(f(r,S)-q \big)^+\Big]+q
$$
and, given a strictly positive constant $\gamma>0$, let 
\begin{equation}
\label{eq:object.L}
	L(r,q):=\widetilde{L}(r,q)+\frac{\gamma}{2}\|q\|^2, \quad \text{and }\overline{L}(r,q) := L(r,q)+\frac{\gamma}{2} \text{dist}^2(r,A)
\end{equation}
be the usual penalized objective function, where 
\begin{equation*}
\text{dist}^2(r,A):=\inf_{x \in A}\|r-x\|^2
\end{equation*}
denotes the squared distance from $r$ to the set $A$.
Since for $\gamma$ small we have
\begin{equation*}
 	\inf_{(r,q) \in \R^d}\overline{L}(r,q)  \approx \overline{\mathrm{AVaR}}_u(f),
 \end{equation*}
we will approximate the left hand side above using the SGLD algorithm, which consists in approximating its minimizer by the (support of the) invariant measure of the Markov chain $(Z^\lambda_{m,h})$ given by
\begin{equation}\label{eq:SGLD}
	Z_{m+1,h}^\lambda = Z_{m,h}^\lambda -\nabla \overline{L} (Z_{m,h}^\lambda)h+\sqrt{2\lambda^{-1}} \xi_m,
\end{equation}
where $(\xi_m)_{m\ge1}$ are independent Gaussian random variables.
In the practice  of financial risk management, the distribution of $S$ is typically unknown.
This is a well-studied issue in quantitative finance, refer for instance to \citep{Roboptim,cont2010robustness,Krae-Sch-Zaeh14} and the references therein.
In particular, $\nabla L$ cannot be directly computed.
It will be replaced by an unbiased estimator.
Following Monte--Carlo simulation ideas, we let $(S^1, \dots, S^P)$ be independent copies of $S$ and $(\widetilde W^1, \dots, \widetilde W^N)$ be independent Brownian motions, and we thus let
\begin{equation*}
	\widetilde{\ell}(z) := \frac{1}{P}\sum_{p=1}^P \frac{1}{1-u}(f(r, S^p) - q)^+ + q, \quad \ell(z):=\widetilde{\ell}(z)+\frac{\gamma}{2}\|q\|^2, \quad \text{and } \overline{\ell}(z):=\ell(z)+\text{dist}^2(r,A), \quad \text{with } z=(r,q) \in \R^d.
\end{equation*}

In the following, we will take $P=N$ for simplicity. Put
\begin{equation}
	\widetilde Z_{m+1,h}^{\lambda,n} = \widetilde Z_{m,h}^{\lambda,n} -\nabla L(\widetilde Z_{m,h}^{\lambda,n})h+\sqrt{2\lambda^{-1}} \Delta \widetilde W_h^n, \quad \text{with} \quad \Delta \widetilde W_h^n: = \widetilde W^n_{m+1}- \widetilde W^n_m,
\end{equation} 
\begin{equation}\label{Milstein}
	\widetilde Z_{m+1,h}^{\prime\lambda,n} = \widetilde Z_{m,h}^{\prime\lambda,n} -\nabla \ell(\widetilde Z_{m,h}^{\prime\lambda,n})h+\sqrt{2\lambda^{-1}} \Delta \widetilde W_h^n, \quad \text{with} \quad \Delta \widetilde W_h^n: = \widetilde W^n_{m+1}- \widetilde W^n_m,
\end{equation} 
 and 
\begin{equation}\label{Milstein_overline}
	\overline Z_{m+1,h}^{\lambda,n} = \overline Z_{m,h}^{\lambda,n} -\nabla \overline{\ell}(\overline Z_{m,h}^{\lambda,n})h+\sqrt{2\lambda^{-1}} \Delta \overline W_h^n, \quad \text{with} \quad \Delta \overline W_h^n: = \overline W^n_{m+1}- \overline W^n_m.
\end{equation} 
Hence we will show that 
\begin{equation}
\label{eq:def.avar.tilde}
	\widetilde{\mathrm{AVaR}}_u(f):=\frac{1}{N}\sum_{n=1}^N \overline \ell(\widetilde{Z}^{\prime \lambda,n}_{M,h})
\end{equation}
approximates $\overline{\mathrm{AVaR}}_u(f)$. Note that the optimal portfolio $r$ can be easily recovered. It is simply the last $d-1$ coordinates of $\overline{Z}^{\lambda,n}_{M,h}$.
Similarly, the value-at-risk can be obtained from the Markov chain $\overline{Z}^{\lambda,n}_{M,h}$.
See \Cref{rem.VaR} for details.
Let us now formulate the assumptions we make on $f$ and the random vector $S$.
\begin{assumption}
\label{ass.f.S}
	The random variable $S$ takes values in $\R^d$ and
	the function $f:\R^d\times R^{d-1}\to \R$ is Borel measurable, and they satisfy\\
	$(i)$ $S$ has finite fourth moment.\\
	$(ii)$ The function $(r,s)\mapsto f(r,s)$ Lipschitz--continuous and continuously differentiable.\\
	$(iii)$ $\inf_r\E[f(r,S)]>0$.\\
	$(iv)$ The random variable $\nabla_rf(r,S)$ is bounded, uniformly in $S$, and $\nabla_s f (r,\cdot)$ is Lipschitz, uniformly in $r$ \label{bounded_gradient_f}\\
	$(v)$ Consider the function $\kappa$ defined as
	\begin{equation*}
		\kappa(u) := \inf\Big\{-\sqrt{2\lambda}\frac{(z-z')\cdot(\nabla L(z') - \nabla L(z))}{\|z - z'\|}, \,\, z,z' \in \R^{d}: \|z -z'\| = u \Big\}.
	\end{equation*}
	It holds
	\begin{equation*}
		\liminf_{u\to \infty}\kappa(u)>0 \quad \text{and}\quad \int_0^1u\kappa(u)^-\diff u<\infty.
	\end{equation*}
\end{assumption}
Let us briefly comment on these conditions before stating the result.
The integrability, regularity, and lower boundedness conditions $(i)-(iii)$ allow to ensure that the problem is well-posed.
The boundedness condition $(iv)$ is assumed mostly to simplify the exposition.
Most of our statements will remain true if it is replaced by a suitable integrability condition.
We introduce the more involved condition $(v)$ to make for the possible lack of convexity of the objective function $L$.
This condition is by now standard when employing coupling by reflection techniques to prove contractivity of diffusion semigroups.
We refer for instance to \citeauthor*{EberlePTRF} \cite{EberlePTRF,EberleCRAS} or the earlier work of \citeauthor*{Chen-Li89} \cite{Chen-Li89}.
Note, for instance, that this condition is automatically satisfied if $f$ is convex (since in this case $L$ is strongly convex) or when $L$ is strictly convex outside a given ball (see \cite[Example 1]{EberlePTRF}).

\medskip

The following is the first main result of this work:
\begin{theorem}
\label{thm:avar}
	Let Assumptions 2.2 hold. Let $t,M,h>0$ be such that $h=\frac{t}{M^2}$. 
	For all $t,\lambda>0,0<\gamma<1,$ and  $M,N\in \mathbb{N}^\star$, we have
	\begin{align}
		&\E\bigg[\Big| \frac{1}{N}\sum_{n=1}^N \overline{\ell}(\widetilde{Z}_{M,h}^{\prime\lambda,n})-\overline{\text{AVaR}}_u(f)\Big|^2\bigg]
		\le C^1_{(u,t,\lambda,t)}\frac{1}{N}+C^2_{(u,t,\lambda)}\gamma^2
		+C^3_{(u,t,\lambda)}h^2+C^4_{(u,\lambda)}e^{- tC^5_{(\lambda)}}+C^6_{(u)}\frac{1}{\lambda^2},
	\end{align}
	where the constants are given in the appendix.
\end{theorem}

Theorem \ref{thm:avar} provides a non-asymptotic rate for the convergence of the estimator $\frac{1}{N}\sum_{n=1}^N \overline{\ell}(\widetilde{Z}^{\prime\lambda,n}_{M,h})$ to the (optimized) average value at risk.
Such a rate is crucial in applications since it gives a precise order of magnitude for the choice of the parameters $M, N, \gamma$ and $\lambda$ needed to achieve a desired order of accuracy.
Moreover, the rate is independent of the dimension $d$, implying in particular that the rate is not made worst when increasing the size of the portfolio $S = (S^1,\dots, S^d)$ (or in general the number of risk factors).
Furthermore, observe that this estimator $\widetilde{\mathrm{AVaR}}_u(f)$ is rather easy to simulate: one only needs to simulate $N$ independent Gaussian random variables, for each of them simulate the iterative scheme \eqref{Milstein} and compute the empirical average of the outcomes. We provide numerical results on the estimation of AVaR in \Cref{simulation} below.

\begin{remark}
	Observe that the method developed here can also allow (with minor changes) to simulate the value function of utility maximization problems of the form
	\begin{equation*}
		\sup_{r\in A}\E^\mu[U(f(r,S))]
	\end{equation*}
	where $U$ is a concave utility function and $\E^\mu$ the expectation when $S\sim \mu$, or even of the robust utility maximization problem
	\begin{equation*}
		\sup_{r\in A}\inf_{\mu \in \mathcal{P}}\E^\mu[U(f(r,S))]
	\end{equation*}
	where $\mathcal{P}$ is the set of possible distributions of $S$.
	In the latter case, one will need to compute (or find an appropriate unbiased estimator of) $\nabla L$ with
	\begin{equation*}
		L(r) := \inf_{\mu \in \mathcal{P}}\E^\mu[U(f(r,S))].
	\end{equation*}
	This is easily done for instance when $\mathcal{P}$ is a ball with respect to the Wasserstein metric around a given distribution $\mu_0$, see e.g. \cite{Roboptim}.
\end{remark}

\subsection{Approximation of general convex risk measures}
\label{subsec.appro.gen.rm}
Let us return to the problem of approximating general law-invariant convex risk measures. 
In this context (as in the case of $\mathrm{AVaR}$) our goal is to simulate the optimized risk measure
\begin{equation*}
	\overline{\rho}(f) := \inf_{r \in A} \rho(f(r,S)).
\end{equation*}
To that end, let us recall a notion of regularity of risk measures introduced in \citep{Ludo} that will be needed to derive an explicit non-asymptotic convergence rate.
Recall that a random variable $X^*$ is said to follow the Pareto distribution with scale parameter $x>0$ and shape parameter $q>0$ if
\begin{equation*}
	P(X\ge t) = \begin{cases}
		(x/t)^q \text{ if } t\ge x\\
		1 \text{ if } t<x.
	\end{cases}
\end{equation*}
\begin{definition}\citep{Ludo}
	Let $q \in (1,\infty)$, and let $X^*$ follow Pareto distribution with scale parameter 1 and shape parameter $q$. A convex risk measure $\rho: \L^{\infty} \to \R$ is said to be $q$-regular if it satisfies
	\begin{equation*}
		\sup_{n \in \N}\rho(X^*\wedge n)<\infty.
	\end{equation*}
\end{definition}
We refer to \citep{Ludo} for a discussion on this notion of regularity, but note for instance that $\mathrm{AVaR}$ is $q$--regular for all $q>1$ and that this notion of regularity is slightly stronger than the well-known Fatou property and the Lebesgue property often assumed for risk measures, see e.g.\ \citeauthor*{foellmer01} \cite{foellmer01}.
Moreover,
one consequence of $q$--regularity is the following slight refinement of Kusuoka's representation: The risk measure $\rho$ satisfies
\begin{equation}\label{bounded}
	\rho(f(r,S))=\sup_{\gamma\in \mathcal{M}:s.t. \beta(\gamma)\leq b}\left(\int_{[0,1)}\text{AVaR}_u(f(r,S))\gamma(du)-\beta(\gamma)\right).
\end{equation}
for some $b>0$, see \citep[Lemma 4.4]{Ludo} for details. 
Thus, the estimator we consider for $\overline{\rho}$ is given by
\begin{equation}\label{general}
	\widetilde{\rho}^\delta(f) :=\esssup_{\gamma \in \mathcal{M}: \beta(\gamma)\leq b} \left(\int_{[0,\delta)}\widetilde{\text{AVaR}_u}(f)\gamma(du)-\beta(\gamma)\right)
\end{equation}
for some $\delta\in (0,1)$, and where $\widetilde{\text{AVaR}_u}(f)$ is the estimator of $\overline{\mathrm{AVaR}_u}(f)$ given by \eqref{eq:def.avar.tilde}, which implicitly depends on $u$ through the objective functions $L$ and $\overline{\ell}$.
The following theorem gives a convergence rate for the approximation of the general law-invariant convex risk measure $\bar \rho(f)$ by $\widetilde{\rho}^\delta(f)$.
\begin{theorem} 
\label{thm:general.rm}
	Let $\rho$ be a $q$--regular convex risk measure with $q>1$.
	Let $f$ be bounded and satisfy the assumptions of \Cref{thm:avar}.
	Let $h=\frac{t}{M^2}$. For all $t,\lambda>0,0<\gamma<1,$ and  $M,N\in \mathbb{N}^\star$, we have
	\begin{align*}
		&\E\Big[|\bar \rho(f)-\widetilde{\rho}^{\delta}(f)|^2] 
		\leq C(1-\delta)^{1/q}+C^7_{(\delta,t,\lambda)}\frac{1}{N}+C^2_{(\delta,t,\lambda)}\gamma^2
		+C^3_{(\delta,t,\lambda)}h^2+C^4_{(\delta,\lambda)}e^{- tC^5_{(\lambda)}}+C^6_{(\delta)}\frac{1}{\lambda^2},
	\end{align*}
	where $C^7_{(\delta,t,\lambda,t)}$  is given in the Appendix, and constants $C^2_{(\delta,t,\lambda)}, C^3_{(\delta,t,\lambda)},C^4_{(\delta,\lambda)}$ and $C^6_{(\delta)}$ correspond to those given in the Appendix, with $u$ replaced by $\delta$. 
	\end{theorem}

\subsection{Numerical results on AVaR}\label{simulation}
Let us complement the above theoretical guarantees with empirical experiments\footnote{Code available at https://github.com/jiaruic/sgld\_risk\_measures}.
We first focus on the approximation of the average value at risk and the value at risk with respect to the time evolution of the Markov chain in the SGLD algorithm.
Thus, for the numerical computations, we set 
$$A=[0,1]^d, \quad \lambda=10^8,\quad \gamma=10^{-8},\quad h=10^{-4},\quad \text{and}\quad  u=0.95.$$
We will consider two cases in our experiments. 
In the first case we assume the underlying distribution to be known and use Monte--Carlo simulation, and in the second case we use real historical stock price data.
\subsubsection{Monte Carlo simulation}
For the Monte Carlo experiments, we set $N=5000$.
\Cref{avar_1d} shows the convergence of AVaR in the 1 dimensional case with $f(r, S)=S$, where $S$ is sampled from a Gaussian distributions. 
\Cref{avar_error} shows the estimation error, $\widetilde{\text{AVaR}}_u-\overline{\text{AVaR}}_u$, where $\overline{\text{AVaR}}_u$ is the theoretical average value at risk for 1-dimensional Gaussian distributions given by
\begin{equation}
\overline{\text{AVaR}}_u=\mu+\sigma\frac{\phi(\Phi^{-1}(u))}{1-u},
\end{equation}
where $\phi$ and $\Phi$ are respectively, the  PDF and the CDF of a standard Gaussian distribution. 
\begin{figure}[h]
\centering
\begin{subfigure}{0.45\textwidth}
  \centering
  \includegraphics[width=1\linewidth]{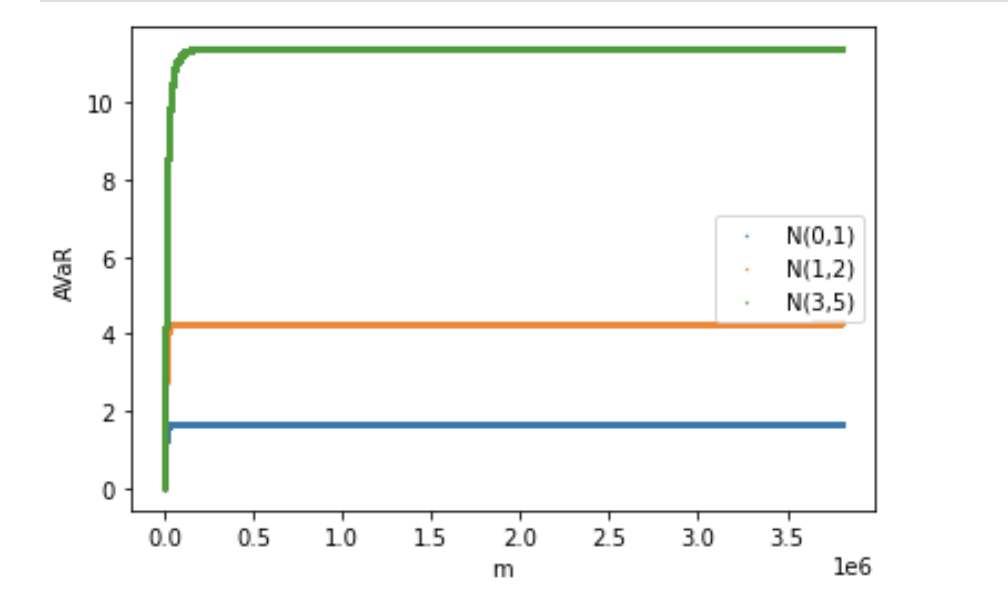}
  \caption{Paths of AVaR}
  \label{avar_1d}
\end{subfigure}%
\hspace{28pt}
\begin{subfigure}{0.45\textwidth}
  \centering
  \includegraphics[width=1\linewidth]{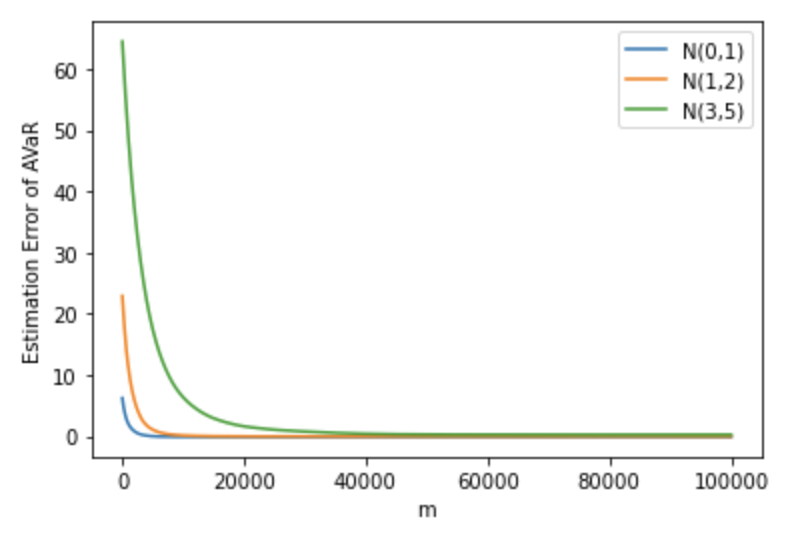}
  \caption{Estimation Error of AVaR}
  \label{avar_error}
\end{subfigure}
\caption{Numerical Results for 1-d Gaussian random variables}
\end{figure}
For the multi-dimensional case, we take the function $$f(r,S)= \sum_{i=1}^d \frac{e^{r_i}}{\sum_{j=1}^d e^{r_j}}S_i, \text{ for } i=1,\cdots,d. $$
\Cref{var_2d} and \Cref{avar_2d} show the convergence of VaR and AVaR in the 2-dimensional case, where $S^1$ is sampled from $\mathcal{N}(1,4)$ and $S^2$ is sampled from $\mathcal{N}(0,1)$. 
\begin{figure}[h]
\centering
\begin{subfigure}{0.45\textwidth}
  \centering
  \includegraphics[width=1\linewidth]{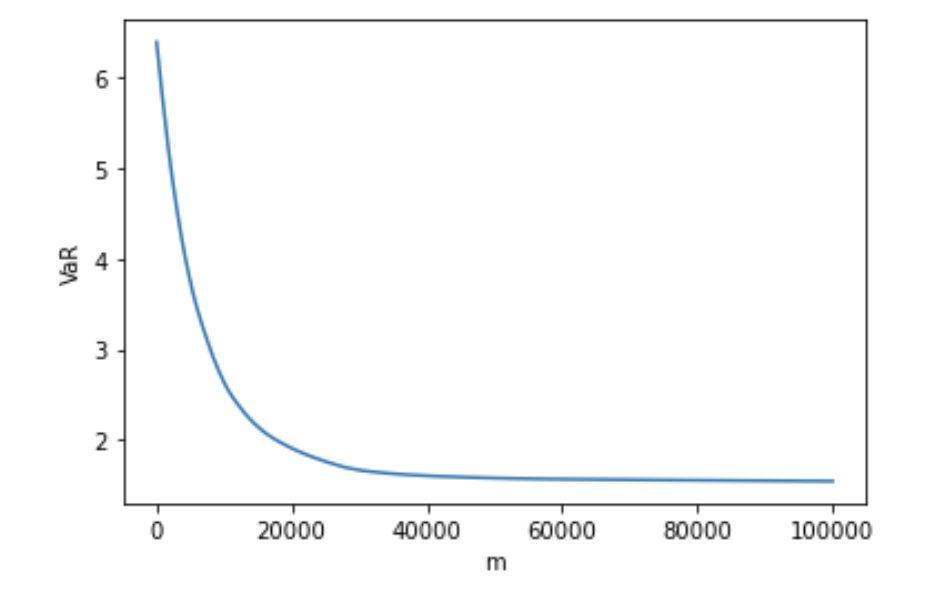}
  \caption{Path of VaR}
  \label{var_2d}
\end{subfigure}%
\hspace{28pt}
\begin{subfigure}{0.45\textwidth}
  \centering
  \includegraphics[width=1\linewidth]{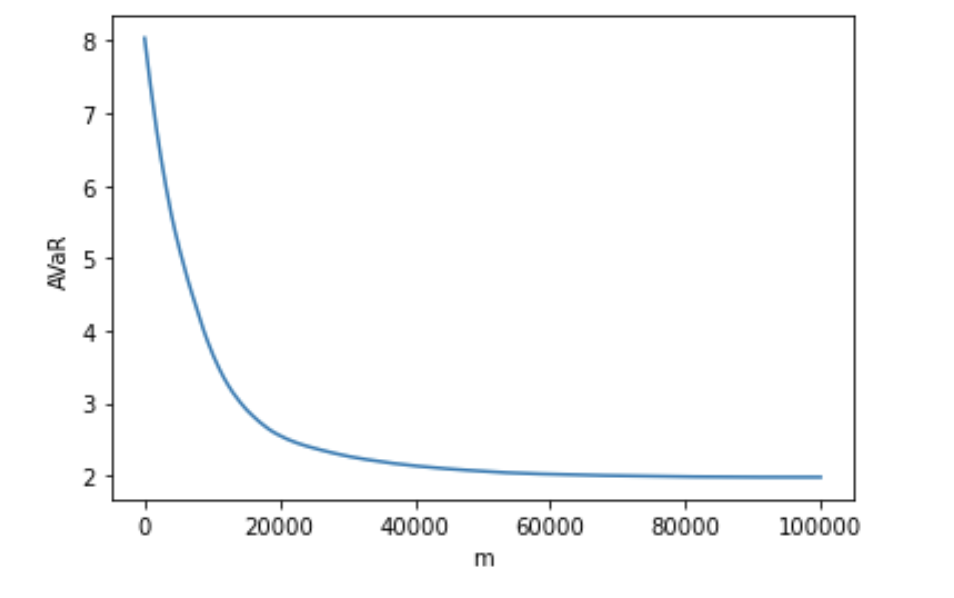}
  \caption{Path of AVaR}
  \label{avar_2d}
\end{subfigure}
\caption{Portfolio of two Gaussian random variables}
\end{figure}

\subsubsection{Numerical results with real data}
In this subsection, we compute AVaR for a portfolio of 106 stocks using real aggregated stock prices over 15-minute time intervals from January 2, 2015 to August 31, 2015.  Among 128 NASDAQ stocks that are "sufficiently liquid", we remove the ones with missing values, and use the remaining 106 stocks. For a detailed description of the data used and for a definition of "sufficiently liquid", please refer to Section 3.2 of \citeauthor*{data} \cite{data}. We use changes in stock prices instead of stock prices themselves, because stock prices are highly dependent. We present paths of the estimated optimized VaR and AVaR of the portfolio of 106 stocks in Figures \ref{var_data} and \ref{avar_data} respectively.

\begin{figure}[h]
\centering
\begin{subfigure}{0.45\textwidth}
  \centering
  \includegraphics[width=1\linewidth]{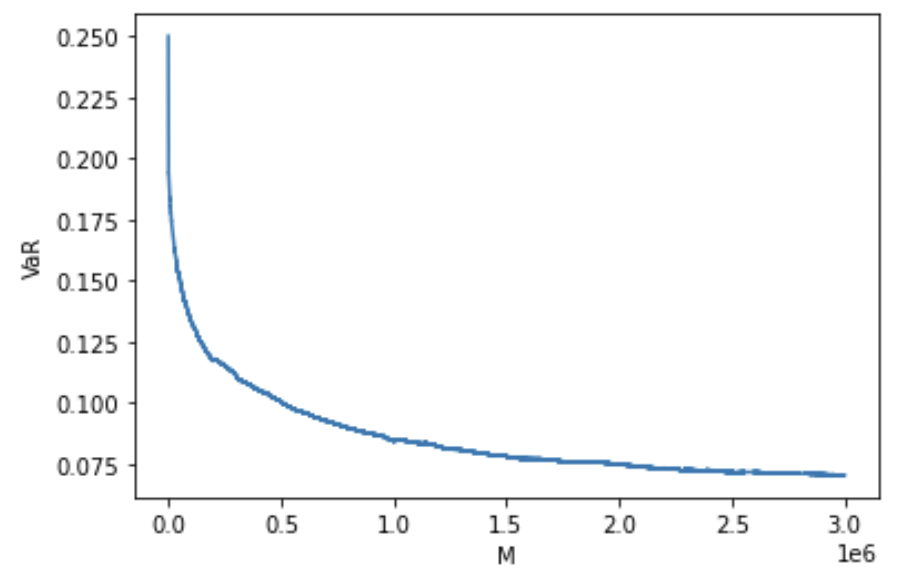}
  \caption{Paths of VaR}
  \label{var_data}
\end{subfigure}%
\hspace{28pt}
\begin{subfigure}{0.45\textwidth}
  \centering
  \includegraphics[width=1\linewidth]{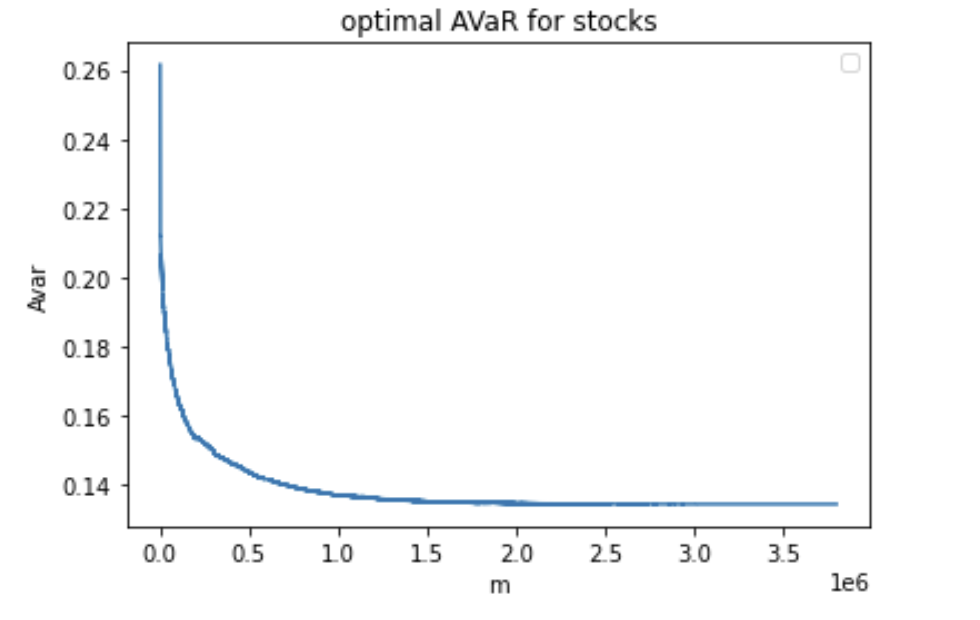}
  \caption{Paths of AVaR}
  \label{avar_data}
\end{subfigure}
\caption{Computations on real price increments of portfolio of 106 stocks. }
\end{figure}

In addition, our approach can also be easily applied to a fixed portfolio of stocks. We take 20 stocks from the 106 stocks described above, and consider a fixed portfolio of equal weights, i.e., $r_i=\frac{1}{20}$ for each $i$. We present paths the estimated  AVaR in Figure \ref{avar_fixed_dist}.
\begin{figure}[h]
\centering
  \includegraphics[width=100mm,scale=0.8]{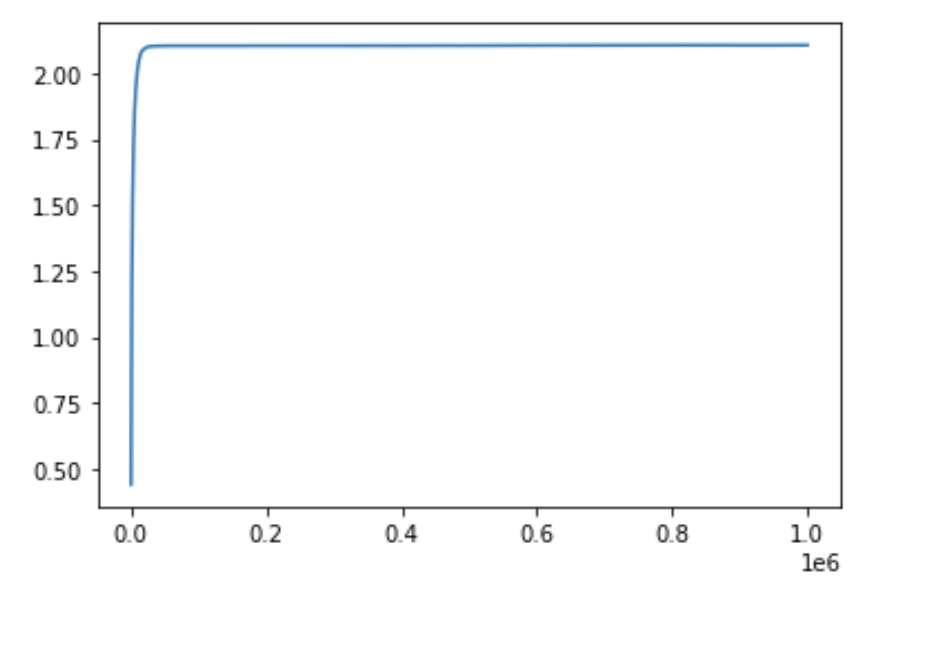}
  \caption{Path of AVaR for a fixed portfolio}
  \label{avar_fixed_dist}
\end{figure}

\subsection{Numerical results on general risk measures}
In order to simulate general risk measures one needs to specify the penalty function $\beta$, or alternatively the precise form of $\rho$ since $\beta$ is given by \cite{follmer}
\begin{equation*}
	\beta(\gamma) = \sup_{X \in \mathcal{A}_\rho}\int_{[0,1]}\mathrm{AVaR}_u(X)\gamma(du),\quad \text{with}\quad \mathcal{A}_\rho:= \{X \in \L^\infty: \rho(X + m)\le 0\}.
\end{equation*}
In general, the simulation of $\tilde\rho^\delta(f)$ as given in \eqref{general} will probably require introducing neural networks since it is the value of an infinite dimensional optimization problem.
This will be addressed in future research.
We will focus here on a case where the problem can be simplified.

In fact, denote by $\partial_\mu\beta$ the so--called linear functional derivative of $\beta$.
It is defined as the function $\partial_\mu\beta: \mathcal{M}([0,1])\times [0,1]\to \R$ such that
\begin{equation*}
	\beta(\mu') - \beta(\mu) = \int_0^1\int_0^1\partial_\mu\beta((1-\lambda)\mu + \lambda\mu',x)(\mu' - \mu)(\d x)\d \lambda.
\end{equation*}
Up to an additive constant, there exists a unique such derivative $\partial_\mu\beta$, see e.g. \cite{MR3752669}.
We have the following:
\begin{proposition}
	Let the assumptions of Theorem \ref{thm:general.rm} hold and assume that $\beta$ admits a second order linear functional derivative that is jointly continuous and such that
	\begin{equation*}
		\sup_{\eta_1, \eta_2 \in \L^2}\E\Big[\sup_{\mu \in \mathcal{M}([0,1])}|\partial_\mu\beta(\mu,\eta_1)| + \sup_{\mu \in \mathcal{M}([0,1])}|\partial_\mu^2\beta(\mu, \eta_1,\eta_2)| \Big]\leq K<\infty.
	\end{equation*}
	Then, it holds
	\begin{align*}
		&\E\bigg[\Big|\inf_{(x_i)_{i=1,\dots,J}\subset [0,1]}F\Big(\frac1J\sum_{i=1}^J\delta_{x_i}\Big) - \rho(f)  \Big|^2\bigg] \\
  \le& \frac{4K^2}{J^2} + C(1-\delta)^{1/q}+C(1-\delta)^{1/q}+C^7_{(\delta,t,\lambda)}\frac{1}{N}+C^2_{(\delta,t,\lambda)}\gamma^2
		+C^3_{(\delta,t,\lambda)}h^2+C^4_{(\delta,\lambda)}e^{- tC^5_{(\lambda)}}+C^6_{(\delta)}\frac{1}{\lambda^2}, 
	\end{align*}
	with $F(\mu):= \int_0^1\widetilde{\mathrm{AVaR}}_u(f)\mu (\d u) - \beta(\mu)$ (recall Equation \eqref{eq:def.avar.tilde}).
\end{proposition}
This is a direct consequence of Theorem \ref{thm:general.rm} and \cite[Theorem 2.4]{Hu-Siska-Szpruch}. The proof is omitted.

\medskip

As an illustrative example, we consider the so--called entropic value-at-risk introduced by \citeauthor*{Ahmadi-Javid} \cite{Ahmadi-Javid} and studied e.g.\ by \citeauthor*{Pichler-Schlotter20} \cite{Pichler-Schlotter20} and \citeauthor*{Follmer-Knispel11} \cite{Follmer-Knispel11} in connection to large portfolio asymptotics.
This is a risk measure based on the R\'enyi entropy given by
\begin{equation*}
	\rho_u(X) = \sup\Big\{\E[ZX]: Z\ge 0, \,\, \E[Z] = 1,\,\, H_q(Z)\le \log\frac{1}{1-u} \Big\}
\end{equation*}
with $H_q(Z):=\frac{1}{q-1}\log \E Z^q$ for $q\in \R^*_+\setminus\{1\}$.
The associated penalty function takes the form
\begin{equation*}
	\beta(\gamma) := \begin{cases}
		0\text{ if } \int_0^1\sigma_\gamma(x)^q\d x\le \big(\frac{1}{1-u}\big)^{q-1}\\
		+\infty \text{ else}
	\end{cases}, \quad \text{with}\quad \sigma_\gamma(x):= \int_0^x\frac{1}{1-v}\gamma(\d v).
\end{equation*}

To numerically compute entropic value-at-risk, for a large $k$, we simulate
$$\sup_{(x_i)_{i=1,\dots,J}\subset [0,1]}F\Big(\frac1J\sum_{i=1}^J\delta_{x_j}\Big):= \sup_{(x_i)_{i=1,\dots,J}\subset [0,1]} \bigg(\frac1J\sum_{i=1}^J\widetilde{\mathrm{AVaR}}_{x_i}(f) - k\Big\{\int_0^1\Big(\frac1J\sum_{i=1}^N\frac{1}{1-x_i}1_{x_j\le x}\Big)^q\d x - (\frac{1}{1-u})^{q-1} \Big\}^+ \bigg).$$
For Monte Carlo simulation, we set $J=5000$, $k=10^{18}$, $q=1.00001$, and estimate the supremum over $(x_i)_{i=1,\dots,J}\subset [0,1]$ by the maximum of 5000 random partitions, each consisting of $J$ points, of the interval $[0,1]$. \Cref{evar} shows the convergence of the entropic value at risk in the 1 dimensional case with $f(r,S)=S$, where $S$ is sampled from $\mathcal{N}(1,2)$. \Cref{evar_error} shows the estimation error compared to the theoretical entropic value-at-risk for  $\mathcal{N}(1,2)$ given by $\rho_u(X) = 1 + \sqrt{-2\log((1-u)2)}$.

\begin{figure}[h]
\centering
\begin{subfigure}{0.45\textwidth}
  \centering
  \includegraphics[width=1\linewidth]{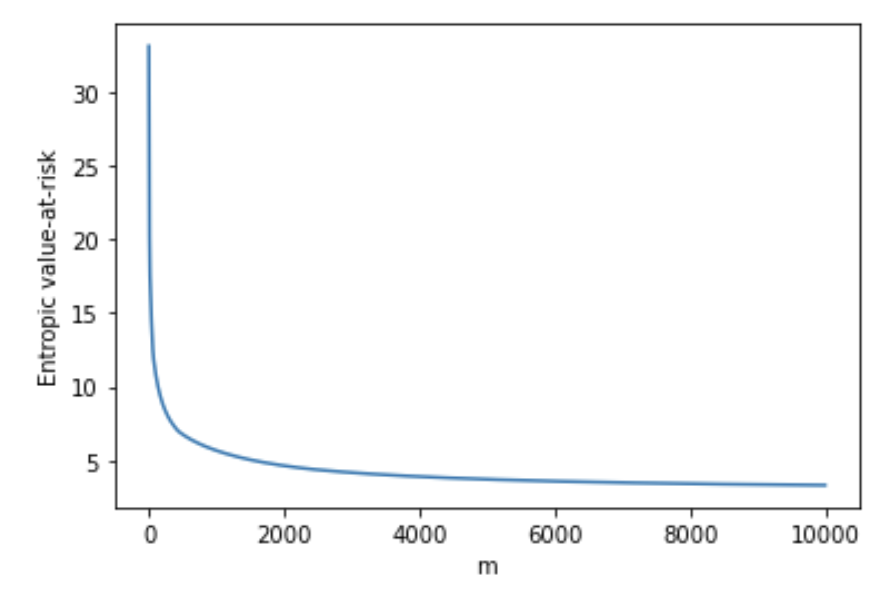}
  \caption{Paths of Entropic Value-at-risk}
  \label{evar}
\end{subfigure}%
\hspace{28pt}
\begin{subfigure}{0.45\textwidth}
  \centering
  \includegraphics[width=1\linewidth]{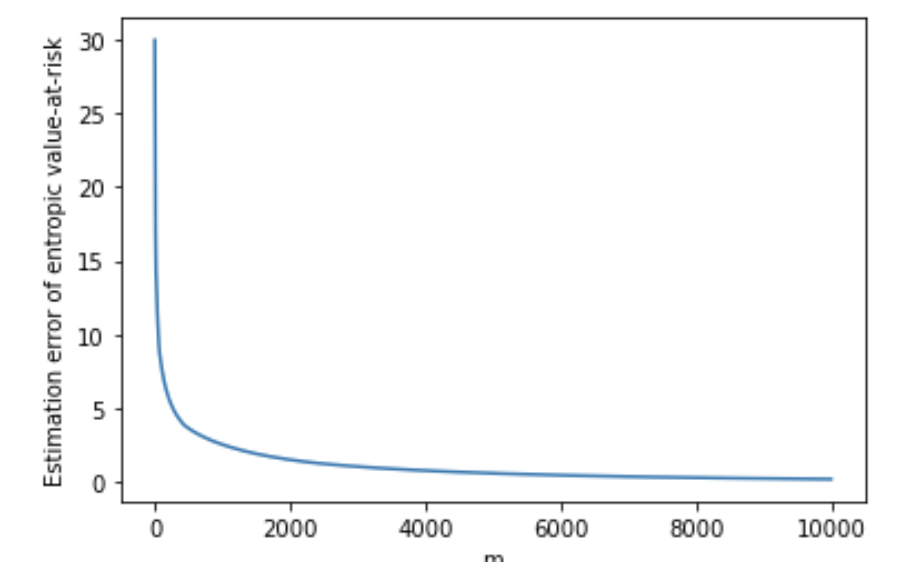}
  \caption{Estimation Error of Entropic Value-at-risk}
  \label{evar_error}
\end{subfigure}
\caption{Numerical Results Entropic Value-at-risk}
\end{figure}

\section{Rates for the average value at risk}
\label{sec:proofs.avar}
This section is dedicated to the proof of Theorem \ref{thm:avar}.
We will start by some preliminary considerations allowing us to introduce ideas used in the proofs.
The details of the proofs will be given in the subsection \ref{sec:proof.thm.avar}.

\subsection{Preliminaries}
The starting point of our method is to recognize \eqref{Milstein} as the $m$--th step of the Euler-Maruyama scheme that discretizes the stochastic differential equation
\begin{equation} \label{SDE}
	\diff Z_t^\lambda=-\nabla {L}(Z_t^\lambda)\diff t+\sqrt{2\lambda^{-1}}\diff W_t,\quad Z_0^\lambda=z\in \mathbb{R}^{d}
\end{equation}
where $W$ is a $d$-dimensional Brownian motion. 
This SDE is the Langevin SDE, with inverse temperature parameter $\lambda$. 
The Langevin SDE is widely studied in physics \citep{sekimoto} and for the sampling of Gibbs distribution via Markov chain Monte--Carlo methods \citep{lmc2}. 
Equipping the probability space $(\Omega, \mathcal{F},\PP)$ with the $\PP$--completion of the filtration of $W$, the equation \eqref{SDE} admits a unique strong solution.
It is well-known that this solution has a unique invariant measure (that we denote by $\mu_\infty^\lambda$) and whose density reads\footnote{As usual, in this article, we use the same notation for a probability measure on $\R^n$ for any $n \in \N$ and its density function.}
\begin{equation}
	\mu_\infty^\lambda(x) =\frac{e^{-\lambda {L}(x)}} {\int_{\R^{d}} e^{-\lambda {L}(z)}\diff z},  
\end{equation}
see e.g.\ \citep[Lemma 2.1]{Lacker}.
In this work, the interest of the Langevin equation (aside from its analytical tractability) stems from the fact that the limiting measure $\mu_\infty$ of $\mu_\infty^\lambda$ as $\lambda \to \infty$ concentrates on the minimizers of $L$, which we will show exist.
This follows from results of Hwang \cite{Hwang}. Intuitively, this means that if $ (r^*,m^*)$ is the minimizer of $L$, then for $\lambda \to \infty$ 
\begin{equation}
\label{eq:Huang.min}
	\int_{\R^{d}} L(z)\mu_\infty^\lambda(\diff z)\approx L(r^*,q^*).
\end{equation}
Moreover, the Langevin equation allows us to exploit classical techniques in order to derive explicit convergence rates to the invariant measure in the present non--convex potential case.
\begin{remark}
\label{rem.VaR}
	One interesting byproduct of our method is that, the simulation of $\mathrm{AVaR}$ directly allows to compute the value at risk and the optimal portfolios, as well as deriving non--asymptotic rates.
	Let us illustrate this on the problem of simulation of optimal portfolios $r^*$ in \Cref{opt}.
	As observed above, $\mu^\lambda_\infty$ converges to a measure $\mu$ supported on the optimal portfolios.
	Now, let $G:\R^{d}\to \R$ be a strictly convex function such that the gradient $\nabla G$ is invertible.
	Then, by Taylor's expansion we have
	\begin{equation*}
		G(\overline{Z}^{\lambda,n}_{M,h}) - G(q^*,r^*) \ge\nabla G(K)(\overline{Z}^{\lambda,n}_{M,h} - (q^*,r^*))
	\end{equation*}
	for some random variable $K$,
	showing that
	\begin{equation*}
		\|\widetilde{Z'}^{\lambda,n}_{M,h} - (q^*,r^*)\| \le \|\nabla G(K)^{-1}\| |G(\widetilde{Z'}^{\lambda,n}_{M,h}) - G(q^*,r^*) |.
	\end{equation*}
	Therefore, provided that the inverse of $\nabla G$ does not grow too fast, the argument we give below to prove \Cref{thm:avar} would allow to derive theoretical guarantees for the optimal portfolio as well, replacing $L$ by $G$.
\end{remark}

\subsection{Proof of Theorem \ref{thm:avar}}
\label{sec:proof.thm.avar}

Throughout this section we assume that the assumptions of Theorem \ref{thm:avar} are satisfied.
We split the proof into several intermediate lemmas.
The first one is probably well known, it asserts that the optimization problem defining $\overline{\mathrm{AVaR}}(f)$ admits a solution. 
\begin{lemma}\label{Hwang}
	The function $\widetilde{L}$ defined in \Cref{eq:object.L} admits a minimum. 
\end{lemma}
\begin{proof}
	In \citep[Proposition 2.1]{Hwang}, \citeauthor{Hwang} gives a sufficient condition for $\widetilde{L}$ admitting a minimum: $\{\mu_\infty^{\lambda}\}$ is tight. A sufficient condition for the tightness of $\{\mu_\infty^{\lambda}\}_{\lambda > 0}$ is that there exists $\varepsilon>0$ such that the set $B:=\{z \in \R^{d}:\widetilde{L}(z)\leq \varepsilon\}$ is compact \citep[Proposition 2.3]{Hwang}. The rest of the proof checks the compactness of set $B$ for any $\epsilon>0$.  

	Since $\widetilde{L}$ is continuous, the set $B=\{z\in \R^{d}:\widetilde{L}(z) \leq \varepsilon\}$ is closed as the pre-image of the closed set $(-\infty,\varepsilon]$. 
	In addition, since $\frac{1}{1-u}(x)^+>x$ for $|x|$ large enough, we have that $B=\{z \in \R^{d}: \frac{1}{1-u}\E[(f(r,S)-q)^+]+q\leq\epsilon\}$ is bounded. To see this, assume to the contrary that $B$ is unbounded. Then there exists a sequence $\{z_i\} \in B$ such that $\|z_i\| \to \infty$. Then for the subsequence of $\{z_i\}$ with $\frac{1}{1-u}(f(r,S)-q)^+>f(r,S)$, we have $\widetilde{L}(x_i) \to \infty$, which contradicts $\widetilde{L}(x)\leq \epsilon$.  
	Thus, the set $B$ is bounded, and is therefore compact. 
\end{proof}

To derive the claimed convergence rate, we decompose the expected error into terms that will be handled independently.
First, we will exploit the approximation \eqref{eq:Huang.min}.
Next, using the $N$ independent Brownian motions $\widetilde W^n$ introduced just before \eqref{Milstein}, we construct $N$ i.i.d. copies $\widehat{Z}_t^{\lambda,n}$ of the solution of the Langevin equation as solutions of the SDEs
\begin{equation}
\label{eq:def.Z.lambda}
	\diff\widehat{Z}^{\lambda,n}_t=-\nabla L(\widehat{Z}^{\lambda,n}_t)\diff t+\sqrt{2 \lambda^{-1}}\diff \widetilde{W}^n,\quad n=1,\dots, N.
\end{equation}
Recall that 
\begin{equation*}
	\widetilde Z_{m+1,h}^{\lambda,n} = \widetilde Z_{m,h}^{\lambda,n} -\nabla L(\widetilde Z_{m,h}^{\lambda,n})h+\sqrt{2\lambda^{-1}} \Delta \widetilde W_h^n, \quad \text{with} \quad \Delta \widetilde W_h^n: = \widetilde W^n_{m+1}- \widetilde W^n_m,
\end{equation*} 
\begin{equation*}
	\widetilde Z_{m+1,h}^{\prime\lambda,n} = \widetilde Z_{m,h}^{\prime\lambda,n} -\nabla \ell(\widetilde Z_{m,h}^{\prime\lambda,n})h+\sqrt{2\lambda^{-1}} \Delta \widetilde W_h^n, \quad \text{with} \quad \Delta \widetilde W_h^n: = \widetilde W^n_{m+1}- \widetilde W^n_m,
\end{equation*} 
 and 
\begin{equation*}
	\overline Z_{m+1,h}^{\lambda,n} = \overline Z_{m,h}^{\lambda,n} -\nabla \overline{\ell}(\overline Z_{m,h}^{\lambda,n})h+\sqrt{2\lambda^{-1}} \Delta \overline W_h^n, \quad \text{with} \quad \Delta \overline W_h^n: = \overline W^n_{m+1}- \overline W^n_m.
\end{equation*}
We decompose the error 
as 
\begin{align}
	&\E \bigg[\Big|\frac{1}{N} \sum_{n=1}^N \overline{\ell}(\overline{Z}^{\lambda,n}_{M,h})-\overline{\text{AVAR}}_u(f)\Big|^2\bigg]\nonumber
	\leq 2^6 \E \bigg[\Big|\frac{1}{N} \sum_{n=1}^N \overline{\ell}(\overline{Z}^{\lambda,n}_{M,h})-\frac{1}{N}\sum_{n=1}^N \ell (\widetilde{Z'}^{\lambda,n}_{M,h})\Big|^2\bigg]+
	2^6 \E \bigg[\Big|\frac{1}{N} \sum_{n=1}^N \ell(\widetilde{Z'}^{\lambda,n}_{M,h})-\frac{1}{N}\sum_{n=1}^N {L}(\widetilde{Z}^{\lambda,n}_{M,h})\Big|^2\bigg]\\
	&\qquad+2^6 \E \bigg[\Big|\frac{1}{N} \sum_{n=1}^N {L}(\widetilde{Z}^{\lambda,n}_{M,h})-\frac{1}{N}\sum_{n=1}^N {L}(\widehat{Z}^{\lambda,n}_t)\Big|^2\bigg]\nonumber
	 +2^6\E \bigg[\Big|\frac1N \sum_{n=1}^N {L}(\widehat{Z}^{\lambda,n}_t)-\E[L(Z_t^\lambda)]\Big|^2\bigg]
	+2^6\Big|\E[L(Z^\lambda_t)]-\int_{\R^{d}} {L} \diff \mu^\lambda_\infty\Big|^2\nonumber\\
	&\qquad +2^6\Big|\int_{\R^{d}} {L} \diff \mu_\infty^\lambda-\int_{\R^{d}} \widetilde{L} \diff\mu_\infty^\lambda\Big|^2
	+2^6 \Big|\int_{\R^{d}} \widetilde{L} d\mu_\infty^\lambda-\overline{\text{AVaR}}_u(f)\Big|^2.
	\label{decomp2}
\end{align}
The rest of the proof consists in controlling each term above separately. 

\begin{lemma}\label{lem:dist}

Under the conditions of Theorem 2.3, for all $t,0<\gamma<1, u\in(0,1), N, M \in \N^\star$ and $\lambda>1$, we have 
\begin{align*}
		&\E \bigg[\Big|\frac{1}{N}\sum_{n=1}^N \overline{\ell}(\overline{Z}^{\lambda,n}_{M,h}) - \frac{1}{N} \sum_{n=1}^N \ell(\widetilde{Z'}^{\lambda,n}_{M,h}) \Big|^2\bigg]
		\leq C^{'1}_{(u,t,\lambda)}h^2+C^{'2}_{(u,t,\lambda)}\gamma^2,
	\end{align*}
	where $C^{'1}_{(u,t,\lambda)}$, and $C^{'2}_{(u,t,\lambda)}$ are given in equations \eqref{C'1} and \eqref{C'2}.
\end{lemma}	
\begin{proof}
    Let $\overline{Z}^{\lambda,n,d-1}_{M,h}$ denote the last $d-1$ coordinates of $\overline{Z}^{\lambda,n}_{M,h}$. By the definition of $\overline{l}$ and Jensen's inequality, we have 
	\begin{align}
		&\E \bigg[\Big|\frac{1}{N}\sum_{n=1}^N \overline{\ell}(\overline{Z}^{\lambda,n}_{M,h})-\frac{1}{N} \sum_{n=1}^N \ell(\widetilde{Z'}^{\lambda,n}_{M,h})\Big|^2\bigg]\nonumber\\
		&\quad =\E \bigg[\Big|\frac{1}{N}\sum_{n=1}^N\left\{\ell(\overline{Z}^{\lambda,n}_{M,h})+\frac{\gamma}{2}\text{dist}^2(\overline{Z}^{\lambda,n,d-1}_{M,h},A)-\ell(\widetilde{Z'}^{\lambda,n}_{M,h})\right\}\Big|^2\bigg]\nonumber\\
		&\quad \leq C\bigg(\frac{1}{N}\sum_{n=1}^N \E\Big[\Big(\ell(\overline{Z}^{\lambda,n}_{M,h})-\ell(\widetilde{Z'}^{\lambda,n}_{M,h})\Big)^2\Big]+\frac{\gamma^2}{N}\sum_{n=1}^N \E \text{dist}^4(\overline{Z}^{\lambda,n,d-1}_{M,h},A)\bigg).\label{z_overline}
\end{align}
For the first term in \eqref{z_overline}, using the definition of $\ell$, we have 
\begin{align}\label{l_decomp}
&\E\Big[(\ell\Big(\overline{Z}^{\lambda,n}_{M,h})-\ell(\widetilde{Z'}^{\lambda,n}_{M,h})\Big)^2\Big]\nonumber\\
=& \E\Big[\Big(\widetilde{\ell}(\overline{Z}^{\lambda,n}_{M,h})-\widetilde{\ell}(\widetilde{Z'}^{\lambda,n}_{M,h})\Big)^2\Big]+\frac{\gamma^2}{2}\E\Big[\big\|\overline{Z}^{\lambda,n}_{M,h}-\widetilde{Z'}^{\lambda,n}_{M,h}\big\|^2\Big]\nonumber\\
\leq & \Bigg(\Big(\frac{2}{1-u}\Big)^2+\frac{\gamma^2}{2}\Bigg)\E\Big[\big\|\overline{Z}^{\lambda,n}_{M,h}-\widetilde{Z'}^{\lambda,n}_{M,h}\big\|^2\Big],
\end{align}
where we used that $\widetilde{\ell}$ is $\frac{2}{1-u}$-Lipschitz in the last step.

To control $\E\Big[\big\|\overline{Z}^{\lambda,n}_{M,h}-\widetilde{Z'}^{\lambda,n}_{M,h}\big\|^2\Big]$, using the definitions of $\overline{Z}^{\lambda,n}_{M,h}$ and $\widetilde{Z'}^{\lambda,n}_{M,h}$, we have
\begin{align*}
\E\Big[\big\|\overline{Z}^{\lambda,n}_{M,h}-\widetilde{Z'}^{\lambda,n}_{M,h}\big\|^2\Big]
\leq C\E\Big[\big\|\overline{Z}^{\lambda,n}_{M-1,h}-\widetilde{Z'}^{\lambda,n}_{M-1,h}\big\|^2\Big]+ Ch^2\E\Big[\big\|\nabla \overline{\ell}(\overline{Z}^{\lambda,n}_{M-1,h})-\nabla \ell (\widetilde{Z'}^{\lambda,n}_{M-1,h})\big\|^2\Big].
\end{align*}
Using this recursive relationship, it can be checked by induction that we have 
\begin{equation}
\E\Big[\big\|\overline{Z}^{\lambda,n}_{M,h}-\widetilde{Z'}^{\lambda,n}_{M,h}\big\|^2\Big]
\leq C\sum_{m=0}^{M-1}h^2 \E\Big[\big\|\nabla \overline{\ell}(\overline{Z}^{\lambda,n}_{m,h})-\nabla \ell (\widetilde{Z'}^{\lambda,n}_{m,h})\big\|^2\Big].\label{bartilde_diff}
\end{equation}
Note that the derivative of $\text{dist}^2(r,A)$ is $2(r-A(r)),$ where we denote $A(r):=\argmin{\text{dist}^2(r,A)}$. In addition, $\nabla A(r) = 1$. Therefore, $\nabla \overline{\ell}$ is $\Big(\frac{1}{1-u}+2\Big)-$ Lipschitz. By adding and then subtracting $\nabla \overline{\ell}(\widetilde{Z'}^{\lambda,n}_{m,h})$, and then using the definitions of $\nabla \overline{\ell}$ and $\nabla \ell$, we have 
\begin{align}
&\E\Big[\big\|\nabla \overline{\ell}(\overline{Z}^{\lambda,n}_{m,h})-\nabla \ell (\widetilde{Z'}^{\lambda,n}_{m,h})\big\|^2\Big]\nonumber\\
\leq & C\E\Big[\big\|\nabla \overline{\ell}(\overline{Z}^{\lambda,n}_{m,h})-\nabla \overline{\ell} (\widetilde{Z'}^{\lambda,n}_{m,h})\big\|^2\Big]+C \E\Big[\big\|\nabla \overline{\ell}(\widetilde{Z'}^{\lambda,n}_{m,h})-\nabla \ell (\widetilde{Z'}^{\lambda,n}_{m,h})\big\|^2\Big]\nonumber\\
\leq & \Big(\frac{1}{(1-u)^2}+2\Big)\E \Big [\big\|\overline{Z}^{\lambda,n}_{m,h}-\widetilde{Z'}^{\lambda,n}_{m,h}\big\|^2\Big]+C\E\Big[\big\|\widetilde{Z'}^{\lambda,n}_{m,h}-A(\widetilde{Z'}^{\lambda,n}_{m,h})\big\|^2\Big].\label{gradient_bar}
\end{align}
Combining equations \eqref{bartilde_diff} and \eqref{gradient_bar}, we have 
\begin{equation*}
\E \Big [\big\|\overline{Z}^{\lambda,n}_{M-1,h}-\widetilde{Z'}^{\lambda,n}_{M-1,h}\big\|^2\Big]
\leq Ch^2 \sum_{m=0}^{M-1}\Bigg\{ \Big(\frac{1}{(1-u)^2}+2\Big)\E \Big [\big\|\overline{Z}^{\lambda,n}_{m,h}-\widetilde{Z'}^{\lambda,n}_{m,h}\big\|^2\Big]+\E\Big[\big\|\widetilde{Z'}^{\lambda,n}_{m,h}\big\|^2\Big]+\E\Big[\big\|A(\widetilde{Z'}^{\lambda,n}_{m,h})\big\|^2\Big]\Bigg\}.
\end{equation*}
Using the discrete version of the Gr\"onwall's inequality \citep[Proposition 5]{discretegronwall}, we have 
\begin{equation}\label{bar_tilde}
\E \Big [\big\|\overline{Z}^{\lambda,n}_{M-1,h}-\widetilde{Z'}^{\lambda,n}_{M-1,h}\big\|^2]\leq Ch^2M\bigg(\E \Big[\big\|\widetilde{Z'}^{\lambda,n}_{M,h}\big\|^2\Big]+\E\Big[\big\|A(\widetilde{Z'}^{\lambda,n}_{M,h})\big\|^2\Big]\bigg)\exp\Bigg(C\bigg(\frac{1}{(1-u)^2}+2\bigg)Mh^2\Bigg).
\end{equation}
For the second term in \eqref{z_overline}, we rewrite as 
\begin{align}\label{dist}
\frac{\gamma^2}{N}\sum_{n=1}^N \E \text{dist}^4(\overline{Z}^{\lambda,n,d-1}_{M,h},A)=\frac{\gamma^2}{N}\sum_{n=1}^N \E\Big[\big\|\overline{Z}^{\lambda,n,d-1}_{M,h}-A(\overline{Z}^{\lambda,n,d-1}_{M,h})\big\|^4\Big]\leq \frac{\gamma^2}{N}\sum_{n=1}^N \E\big\|\overline{Z}^{\lambda,n}_{M,h}\big\|^4+C \gamma^2,
\end{align}
where the last step follows because the set $A$ is compact. 

It remains to bound the fourth moments of $\overline{Z}^{\lambda,n}_{M,h}$ and $\widetilde{Z}^{',\lambda,n}_{M,h}$. Using the definition of $\overline{Z}^{\lambda,n}_{m+1,h}$, and letting $C_d=\Big(\frac{d}{2}+1\Big)\Big(\frac{d}{2}\Big)$ we have
	\begin{align*}
		\E[\big\|\overline{Z}^{\lambda,n}_{m,h}\big\|^4]
		&=\E\Big[ \big\|\widetilde{Z}^{\lambda,n}_{m-1,h}-h\nabla \overline{\ell}(\overline{Z}^{\lambda,n}_{m-1,h})+\sqrt{2\lambda^{-1}}\Delta \widetilde{W_h}\big\|^4 \Big]\\	
		& \leq C\bigg(\E\Big[\big\|\overline{Z}^{\lambda,n}_{m-1,h}\big\|^4\Big]+h^4\E \Big[|\nabla \overline{\ell}( \overline{Z}^{\lambda,n}_{m-1,h})|^4\Big]+\frac{h^2}{\lambda^2}\Big(\frac{d}{2}+1\Big)\Big(\frac{d}{2}\Big)\bigg)\\
		& \leq C\bigg(\E\Big[\big\| \overline{Z}^{\lambda,n}_{m-1,h}\big\|^4\Big]+h^4\E \Big[\big|\nabla \ell( \overline{Z}^{\lambda,n}_{m-1,h})\big|^4\Big]+h^4 \E\Big[\big\|2(\overline{Z}^{\lambda,n}_{m-1,h}-A(\overline{Z}^{\lambda,n}_{m-1,h}))\big\|^4\Big]+\frac{h^2}{\lambda^2}C_d\bigg)\\
		&\leq C\bigg(\E\Big[\big\|\overline{Z}^{\lambda,n}_{m-1,h}\big\|^4\Big]+h^4\E\bigg[\Big\|\frac{2}{1-u}+\gamma\overline{Z}^{\lambda,n}_{m-1,h}\Big\|^4\bigg]+h^4\E\Big[\big\|\overline{Z}^{\lambda,n}_{m-1,h}\big\|^4\Big]+h^4\E \Big[\big\|A(\overline{Z}^{\lambda,n}_{m-1,h})\big\|^4\Big]+\frac{h^2}{\lambda^2}C_d\bigg)\\
	&\leq  C\left((1+\gamma^4h^4+h^4)\E\Big[\big\|\overline{Z}^{\lambda,n}_{m-1,h}\big\|^4\Big]+\frac{h^4}{(1-u)^4}+h^4+\frac{h^2}{\lambda^2}C_d\right),
	\end{align*}
	where we used the compactness of $A$ in the last step. Using this recursive relationship, it can be checked by induction that for all $m \in \N$, we have 
	\begin{align}\label{z_barmoment}
	\E \Big[\big\|\overline{Z}^{\lambda,n}_{m,h}\big\|^4\Big] 
	&\leq C\bigg((1+\gamma^4h^4+h^4)^m\E[\big\|\overline{Z}^{\lambda,n}_{0,h}\big\|^4]+\Big(\frac{h^4}{(1-u)^4}+h^4+\frac{h^2}{\lambda^2}C_d\Big)\sum_{i=1}^m(1+\gamma^4h^4+h^4)^i\bigg)\nonumber\\
	&\leq C\bigg((1+\gamma^4h^4+h^4)^m+\Big(\frac{h^4}{(1-u)^4}+h^4+\frac{h^2}{\lambda^2}C_d\Big)\frac{(1+\gamma^4h^4+h^4)((1+\gamma^4h^4+h^4)^m-1)}{(1+\gamma^4h^4+h^4)-1}\bigg)\nonumber\\
	&\leq C\bigg((1+\gamma^4h^4+h^4)^m+\Big(\frac{h^4}{(1-u)^4}+h^4+\frac{h^2}{\lambda^2}C_d\Big)(1+\gamma^4h^4+h^4)^m\bigg),
	\end{align}
	where the second inequality uses sum of geometric series.
	The bound of $\widetilde{Z}^{'\lambda,n}_{M,h}$ is given in the proof of Lemma 3.4 below. Combining equations \eqref{z_overline}, \eqref{l_decomp}, \eqref{bar_tilde},\eqref{dist}, and the moments given in equations \eqref{eq:4.moment} and \eqref{z_barmoment}, and taking $h,\gamma \leq 1$, and $M \geq 1$,  we have the result of the lemma. 
\end{proof}

\begin{lemma}
\label{lem:l.vs.ell}
	Under the conditions of Theorem 2.3, for all $t,\gamma>0, u\in(0,1), N, M \in \N^\star$ and $\lambda>1$ if  $h<\frac{1}{\left(\frac{2}{1-u}+\gamma\right)}$ then we have 
	\begin{align*}
		&\E \bigg[\Big|\frac{1}{N}\sum_{n=1}^N \ell(\widetilde{Z'}^{\lambda,n}_{M,h}) - \frac{1}{N} \sum_{n=1}^N L(\widetilde{Z}^{\lambda,n}_{M,h}) \Big|^2\bigg]
		\le C^{'3}_{(u,t)}h^2+\big(1+C^{'4}_{(u,t)}\big)\frac{C}{N}+C^{'5}_{(u,t,\lambda)}\gamma^2,
	\end{align*}
	for some constants $C^{'3}_{(u,t)}, C^{'4}_{(u,t)}$ and $ C^{'5}_{(u,t,\lambda)}$ given in equations \eqref{C'3} - \eqref{C'5}.
\end{lemma}
\begin{proof}
	By the definition of $L$ and Jensen's inequality, we have 
	\begin{align}
		&\E \bigg[\Big|\frac{1}{N}\sum_{n=1}^N \ell(\widetilde{Z'}^{\lambda,n}_{M,h})-\frac{1}{N} \sum_{n=1}^N L(\widetilde{Z}^{\lambda,n}_{M,h})\Big|^2\bigg]\nonumber\\
		&\quad =\E \Bigg[\bigg|\frac{1}{N}\sum_{n=1}^N\left\{\widetilde{L}(\widetilde{Z}^{\lambda,n}_{M,h})+\frac{\gamma}{2}\big\|\widetilde{Z}^{\lambda,n}_{M,h}\big\|^2-\widetilde{\ell}(\widetilde{Z'}^{\lambda,n}_{M,h})-\frac{\gamma}{2}\big\|\widetilde{Z'}^{\lambda,n}_{M,h}\big\|^2\right\}\bigg|^2\Bigg]\nonumber\\
		&\quad \leq C\bigg(\frac{1}{N}\sum_{n=1}^N \E\Big[(\widetilde{L}(\widetilde{Z}^{\lambda,n}_{M,h})-\widetilde{\ell}\big(\widetilde{Z'}^{\lambda,n}_{M,h})\big)^2\Big]+\frac{\gamma^2}{N}\sum_{n=1}^N \E \Big[\big(\big\|\widetilde{Z}^{\lambda,n}_{M,h}\big\|^2-\big\|\widetilde{Z'}^{\lambda,n}_{M,h}\big\|^2\big)^2\Big]\bigg)\nonumber\\
		&\quad \leq \frac{C}{N}\sum_{n=1}^N\E\Big[\big(\widetilde{L}(\widetilde{Z}^{\lambda,n}_{M,h})-\widetilde{\ell}(\widetilde{Z'}^{\lambda,n}_{M,h})\big)^2\Big]
		+C\frac{\gamma^2}{N}\sum_{n=1}^N\Big(\E\Big[\big\|\widetilde{Z}^{\lambda,n}_{M,h}\big\|^4\Big]+\E\Big[\big\|\widetilde{Z'}^{\lambda,n}_{M,h}\big\|^4\Big]\Big),\label{z'}
\end{align}
where the last inequality follows by Cauchy-Schwarz inequality. 
%%%%%%%%%%%%%%%%%%%%%%%%%%%

%We will show in \Cref{lem:int} below that $\|\widetilde{Z}^{\lambda,n}_{M,h}\|$ have exponential moments bounded by a constant $C_{M,h,\lambda}$, and similar arguments allow to show that $\|\widetilde{Z'}^{\lambda,n}_{M,h}\|^4$ also has exponential moment bounded by the same constant.
%This shows in particular that 
%\begin{equation}
%\label{eq:4.moment}
	%\E\Big[\|\widetilde{Z}^{\lambda,n}_{M,h}\|^4\Big]+\E\Big[\|\widetilde{Z'}^{\lambda,n}_{M,h}\|^4\Big] \le 2C_{(M,t,\lambda)}.
%\end{equation}
%\medskip

Next, we will bound the fourth moments of $\widetilde{Z}^{\lambda,n}_{M,h}$ and $\widetilde{Z'}^{\lambda,n}_{M,h}$.
	Recall $C_d=\Big(\frac{d}{2}+1\Big)\Big(\frac{d}{2}\Big)$. Using the definition of $\widetilde{Z}^{\lambda,n}_{m+1,h}$, we have
	\begin{align*}
		\E\Big[\big\|\widetilde{Z}^{\lambda,n}_{m,h}\big\|^4\Big]
		&=\E\Big[ \big\|\widetilde{Z}^{\lambda,n}_{m-1,h}-h\nabla L(\widetilde{Z}^{\lambda,n}_{m-1,h})+\sqrt{2\lambda^{-1}}\Delta \widetilde{W_h}\big\|^4 \Big]\\	
		& \leq C\bigg(\E\Big[\big\|\widetilde{Z}^{\lambda,n}_{m-1,h}\big\|^4\Big]+h^4\E\Big[|\nabla L( \widetilde{Z}^{\lambda,n}_{m-1,h})|^4\Big]+\frac{h^2}{\lambda^2}C_d\bigg)\\
		&\leq C\bigg(\E\Big[\big\|\widetilde{Z}^{\lambda,n}_{m-1,h}\big\|^4\Big]+h^4\E\bigg[\Big\|\frac{2}{1-u}+\gamma\widetilde{Z}^{\lambda,n}_{m-1,h}\Big\|^4\bigg]+\frac{h^2}{\lambda^2}C_d\bigg)\\
		&\leq  C\left(\E\Big[\big\|\widetilde{Z}^{\lambda,n}_{m-1,h}\big\|^4\Big]+\frac{h^4}{(1-u)^4}+\gamma^4h^4\E\Big[\big\|\widetilde{Z}^{\lambda,n}_{m-1,h}\big\|^4\Big]+\frac{h^2}{\lambda^2}C_d\right)\\
	&\leq  C\left((1+\gamma^4h^4)\E\Big[\big\|\widetilde{Z}^{\lambda,n}_{m-1,h}\big\|^4\Big]+\frac{h^4}{(1-u)^4}+\frac{h^2}{\lambda^2}C_d\right),
	\end{align*}
	where the third inequality follows from the fact that $\widetilde{L}$ is $\frac{2}{1-u}-$Lipschitz.
	Using this recursive relationship, it can be checked by induction that for all $m \in \N$, 
	\begin{align}
	\label{eq:4.moment}
		\E\Big[\big\|\widetilde{Z}^{\lambda,n}_{m,h}\big\|^4\Big]
	&\leq C\bigg((1+\gamma^4h^4)^m\E\Big[\big\|\widetilde{Z}^{\lambda,n}_{0,h}\big\|^4\Big]+\Big(\frac{h^4}{(1-u)^4}+\frac{h^2}{\lambda^2}C_d\Big)\sum_{i=1}^m(1+\gamma^4h^4)^i\bigg)\nonumber\\
	&\leq C\bigg((1+\gamma^4h^4)^m+\Big(\frac{h^4}{(1-u)^4}+\frac{h^2}{\lambda^2}C_d\Big)\frac{(1+\gamma^4h^4)((1+\gamma^4h^4)^m-1)}{(1+\gamma^4h^4)-1}\bigg)\nonumber\\
	&\leq C\bigg((1+\gamma^4h^4)^m+\Big(\frac{h^4}{(1-u)^4}+\frac{h^2}{\lambda^2}C_d\Big)(1+\gamma^4h^4)^m\bigg)
\end{align}
where the second inequality follows by properties of geometric series.
By the same argument, we also have
	\begin{align}
		\E\Big[\big\|\widetilde{Z'}^{\lambda,n}_{M,h}\big\|^4\Big]
		\leq C\bigg((1+\gamma^4h^4)^M+\Big(\frac{h^4}{(1-u)^4}+\frac{h^2}{\lambda^2}C_d\Big)(1+\gamma^4h^4)^M\bigg).\label{z'moment}
	\end{align}
	
	Thus, we have
	\begin{align}
		\frac{\gamma^2}{N}\sum_{n=1}^N\Big(\E\Big[\big\|\widetilde{Z}^{\lambda,n}_{M,h}\big\|^4\Big]+\E\Big[\big\|\widetilde{Z'}^{\lambda,n}_{M,h}\big\|^4\Big]\Big)
		\leq \gamma^2 C^{'5}_{(u,t,\gamma)},
		\label{z'second}
\end{align}
where $C^{'5}_{(u,t,\gamma)}$ is given in \eqref{C'5}.
\medskip

	Let us now turn to the first term on the right hand side of \eqref{z'}. 
	Adding and subtracting $\widetilde{L}(\widetilde{Z'}^{\lambda,n}_{M,h})$, we have 
	\begin{align}
		&\E\Big[ \Big(\widetilde{L}(\widetilde{Z}^{\lambda,n}_{M,h})-\widetilde{\ell}(\widetilde{Z'}^{\lambda,n}_{M,h})\Big)^2 \Big]
		\leq 2\E\Big[\Big(\widetilde{L}(\widetilde{Z}^{\lambda,n}_{M,h})-\widetilde{L}(\widetilde{Z'}^{\lambda,n}_{M,h}) \Big)^2\Big]+2\E\Big[\Big(\widetilde{L}(\widetilde{Z'}^{\lambda,n}_{M,h})-\widetilde{\ell}(\widetilde{Z'}^{\lambda,n}_{M,h}) \Big)^2\Big]\label{z'2}.
	\end{align}
 Using that $\widetilde{L}$ is $\frac{2}{1-u}$--Lipschitz, it follows that
	\begin{align*}
		\E\Big[\Big(\widetilde{L}(\widetilde{Z}^{\lambda,n}_{M,h})-\widetilde{L}(\widetilde{Z'}^{\lambda,n}_{M,h})\Big)^2\Big]
		&\leq \frac{4}{(1-u)^2} \E \left[\big\|\widetilde{Z}^{\lambda,n}_{M,h}-\widetilde{Z'}^{\lambda,n}_{M,h}\big\|^2\right]
  \end{align*}
   Next, we will bound the fourth moment of the difference between $\widetilde{Z}^{\lambda,n}_{M,h}$ and $\widetilde{Z'}^{\lambda,n}_{M,h}$. Using the definitions of $\widetilde{Z}^{\lambda,n}_{M,h}$ and $\widetilde{Z'}^{\lambda,n}_{M,h}$, given in \Cref{eq:SGLD} and \Cref{Milstein} respectively, we have 
	\begin{align*}
		&\E \left[\big\|\widetilde{Z}^{\lambda,n}_{M,h}-\widetilde{Z'}^{\lambda,n}_{M,h}\big\|^4\right]
		=\E\left[\big\|\widetilde{Z}^{\lambda,n}_{M-1,h}-h\nabla L(\widetilde{Z}^{\lambda,n}_{M-1,h})-\widetilde{Z'}^{\lambda,n}_{M-1,h}+h\nabla \ell(\widetilde{Z'}^{\lambda,n}_{M-1,h})\big\|^4\right]\nonumber\\
		\quad &\leq C\E \left[\big\|\widetilde{Z}^{\lambda,n}_{M-1,h}-\widetilde{Z'}^{\lambda,n}_{M-1,h}\big\|^4\right]+Ch^4\E\left[\big\|\nabla L(\widetilde{Z}^{\lambda,n}_{M-1,h})-\nabla \ell(\widetilde{Z'}^{\lambda,n}_{M-1,h})\big\|^4\right].
	\end{align*}
	Using this recursive relationship, it can be checked by induction that we have 
	\begin{equation}
		\E \left[\big\|\widetilde{Z}^{\lambda,n}_{M,h}-\widetilde{Z'}^{\lambda,n}_{M,h}\big\|^4\right]\leq C\sum_{m=0}^{M-1}h^4\E\left[\big\|\nabla L(\widetilde{Z}^{\lambda,n}_{m,h})-\nabla \ell(\widetilde{Z'}^{\lambda,n}_{m,h})\big\|^4\right]. \label{z'fourth}
	\end{equation}
	Let $\widetilde{Z}^{\lambda,n,d-1}_{M,h}$ (resp. $\widetilde{Z'}^{\lambda,n,d-1}_{M,h}$) denote the last $d-1$ coordinates of $\widetilde{Z}^{\lambda,n}_{M,h}$ (resp. $\widetilde{Z'}^{\lambda,n}_{M,h}$).
	Define the vector $\boldsymbol{S} :=(-1,S^1,\cdots,S^d)$, and let $e_1\in \RR^{d}$ be a vector with $1+\gamma m$ in the first entry and $0$ everywhere else. Adding and subtracting $\nabla L(\widetilde{Z'}^{\lambda,n}_{M,h})$, we get
	\begin{align}
		&\E\left[\big\|\nabla L(\widetilde{Z}^{\lambda,n}_{M,h})-\nabla \ell(\widetilde{Z'}^{\lambda,n}_{M,h})\big\|^4\right]\nonumber\\
		\leq& C\E\left[\big\|\nabla L(\widetilde{Z}^{\lambda,n}_{M,h})-\nabla L(\widetilde{Z'}^{\lambda,n}_{M,h})\big\|^4\right]+ C\E\left[\big\|\nabla L(\widetilde{Z'}^{\lambda,n}_{M,h})-\nabla \ell(\widetilde{Z'}^{\lambda,n}_{M,h})\big\|^4\right]\nonumber\\
		\leq& C\left(\frac{1}{1-u}+\gamma\right)^4\E \left[\big\|\widetilde{Z}^{\lambda,n}_{M,h}-\widetilde{Z'}^{\lambda,n}_{M,h}\big\|^4\right]\nonumber\\
		&+C\E\bigg[\bigg|\Big(e_1+\frac{1}{1-u}\E\big[\nabla_r f(\widetilde{Z'}^{\lambda,n,d-1}_{M,h}, S) 1_{\{f(\widetilde{Z'}^{\lambda,n,d-1}_{M,h}, S) \geq \widetilde{Z'}^{\lambda,n}_{M,h}\}}|\widetilde{Z'}^{\lambda,n}_{M,h}\big] \Big)\\
		&- \Big(e_1+\frac{1}{1-u}\frac{1}{N}\sum_{i=1}^N (\nabla f(\widetilde{Z'}^{\lambda,n,d-1}_{M,h}, S^i) 1_{\{f(\widetilde{Z'}^{\lambda,n,d-1}_{M,h}, S^i)\geq\widetilde{Z'}^{\lambda,n}_{M,h}(1)\}})\Big)\bigg|^4\bigg]\nonumber\\
		\leq& C\bigg(\frac{1}{1-u}+\gamma\bigg)^4\E \bigg[\big\|\widetilde{Z}^{\lambda,n}_{M,h}-\widetilde{Z'}^{\lambda,n}_{M,h}\big\|^4\bigg]+C\bigg(\frac{1}{1-u}\bigg)^4,\label{z'gradient}
\end{align}
where the last inequality follows by Assumption \ref{bounded_gradient_f} $(iv)$
.
Putting equations \eqref{z'fourth}, \eqref{z'gradient} together, we have 
\begin{equation*}
\E \left[\big\|\widetilde{Z}^{\lambda,n}_{M,h}-\widetilde{Z'}^{\lambda,n}_{M,h}\big\|^4\right]
\leq Ch^2\left(\frac{1}{1-u}+\gamma\right)^4 \sum_{m=0}^{M-1}\E \left[\big\|\widetilde{Z}^{\lambda,n}_{m,h}-\widetilde{Z'}^{\lambda,n}_{m,h}\big\|^4\right]+C Mh^4\left(\frac{1}{1-u}\right)^4.
\end{equation*}
Using the discrete version of the Gr\"onwall's inequality \citep[Proposition 5]{discretegronwall}, we have 
	\begin{equation}
		\E \left[\big\|\widetilde{Z}^{\lambda,n}_{M,h}-\widetilde{Z'}^{\lambda,n}_{M,h}\big\|^4\right]
		\leq CMh^4\left(\frac{1}{1-u}\right)^4 \exp\left(CMh^4\left(\frac{1}{1-u}+\gamma\right)^4\right).\label{z'diff}
	\end{equation}
	Therefore, it follows that
	\begin{align*}
		\E\Big[\Big(\widetilde{L}(\widetilde{Z}^{\lambda,n}_{M,h})-\widetilde{L}(\widetilde{Z'}^{\lambda,n}_{M,h})\Big)^2\Big]
		&\le \frac{4}{(1-u)^4}M^{1/2}h^2\exp\left(CMh^4\left(\frac{1}{1-u}+\gamma\right)^4\right).
	\end{align*}

	For the second term on the right hand side of \eqref{z'2}, we use a law of large number type argument.
	In fact, we have
	\begin{align}
		&\E\Big[(\widetilde{L}(\widetilde{Z'}^{\lambda,n}_{M,h})-\widetilde{\ell}(\widetilde{Z'}^{\lambda,n}_{M,h}))^2\Big]
		=\E\bigg[ \E\Big[(\widetilde{L}(\widetilde{Z'}^{\lambda,n}_{M,h})-\widetilde{\ell}(\widetilde{Z'}^{\lambda,n}_{M,h}))^2|\widetilde{Z'}^{\lambda,n}_{M,h}\Big] \bigg]\nonumber\\
		\qquad &=\E\left[\E\left[\left(\frac{1}{1-u} \E\Big[(f(\widetilde{Z'}^{\lambda,n,d-1}_{M,h},S)-\widetilde{Z'}^{\lambda,n}_{M,h})^+\Big]-\frac{1}{1-u}\frac{1}{N}\sum_{i=1}^N(f(\widetilde{Z'}^{\lambda,n,d-1}_{M,h}, S^i)-\widetilde{Z'}^{\lambda,n}_{M,h})^+\right)^2\Bigg| \widetilde{Z'}^{\lambda,n}_{M,h}\right]\right]\nonumber\\
		\qquad&= \left(\frac{1}{1-u}\right)^2\frac{1}{N^2}\E\Bigg[\E\Bigg[\sum_{i,j=1}^N\left((f(\widetilde{Z'}^{\lambda,n,d-1}_{M,h}, S^i)-\widetilde{Z'}^{\lambda,n}_{M,h})^+-\E\Big[(f(\widetilde{Z'}^{\lambda,n,d-1}_{M,h},S)-\widetilde{Z'}^{\lambda,n}_{M,h})^+|\widetilde{Z'}^{\lambda,n}_{M,h}\Big]\right)\nonumber\\
		&\qquad\qquad \cdot\left((f(\widetilde{Z'}^{\lambda,n,d-1}_{M,h}, S^j)-\widetilde{Z'}^{\lambda,n}_{M,h})^+- \E\Big[(f(\widetilde{Z'}^{\lambda,n,d-1}_{M,h},S)-\widetilde{Z'}^{\lambda,n}_{M,h}(1))^+|\widetilde{Z'}^{\lambda,n}_{M,h}\Big]\right)\Big|\widetilde{Z'}^{\lambda,n}_{M,h}\Bigg]\Bigg]
	\end{align}
	For $i \neq j$, using that $S^i$ and $S^j$ are independent, we have
	\begin{align*}
		&\E\Bigg[\E\Bigg[\left((f(\widetilde{Z'}^{\lambda,n,d-1}_{M,h}, S^i)-\widetilde{Z'}^{\lambda,n}_{M,h})^+-\E\Big[f(\widetilde{Z'}^{\lambda,n,d-1}_{M,h}, S)-\widetilde{Z'}^{\lambda,n}_{M,h})^+|\widetilde{Z'}^{\lambda,n}_{M,h}\Big]\right)\\
		&\cdot\left((f(\widetilde{Z'}^{\lambda,n,d-1}_{M,h}, S^j)-\widetilde{Z'}^{\lambda,n}_{M,h})^+- \E\Big[(f(\widetilde{Z'}^{\lambda,n,d-1}_{M,h}, S)-\widetilde{Z'}^{\lambda,n}_{M,h})^+|\widetilde{Z'}^{\lambda,n}_{M,h}\Big]\right)\Big|\widetilde{Z'}^{\lambda,n}_{M,h}\Big]\Bigg]\\
		=&\E\Bigg[\E\Big[(f(\widetilde{Z'}^{\lambda,n,d-1}_{M,h}, S^i)-\widetilde{Z'}^{\lambda,n}_{M,h})^+-\E\Big[(f(\widetilde{Z'}^{\lambda,n,d-1}_{M,h}, S)-\widetilde{Z'}^{\lambda,n}_{M,h}(1))^+|\widetilde{Z'}^{\lambda,n}_{M,h}\Big]\Big|\widetilde{Z'}^{\lambda,n}_{M,h}\Big]\Bigg]\\
		&\cdot\E\Bigg[\E\Big[(f(\widetilde{Z'}^{\lambda,n,d-1}_{M,h}, S^j)-\widetilde{Z'}^{\lambda,n}_{M,h}(1))^+-\E\Big[(f(\widetilde{Z'}^{\lambda,n,d-1}_{M,h}, S)-\widetilde{Z'}^{\lambda,n}_{M,h})^+|\widetilde{Z'}^{\lambda,n}_{M,h}\Big]\Big|\widetilde{Z'}^{\lambda,n}_{M,h}\Big]\Bigg]
	=0.
	\end{align*}
	Therefore, we can estimate the desired term as
	\begin{align}
		&\E\Big[\Big(\widetilde{L}(\widetilde{Z'}^{\lambda,n}_{M,h})-\widetilde{\ell}(\widetilde{Z'}^{\lambda,n}_{M,h})\Big)^2\Big]\\ &=\left(\frac{1}{1-u}\right)^2\frac{1}{N^2}\sum_{i=1}^N\E\Bigg[\E\Big[\left((f(\widetilde{Z'}^{\lambda,n,d-1}_{M,h}, S^i)-\widetilde{Z'}^{\lambda,n}_{M,h})^+-\E\Big[(f(\widetilde{Z'}^{\lambda,n,d-1}_{M,h}, S)-\widetilde{Z'}^{\lambda,n}_{M,h})^+|\widetilde{Z'}^{\lambda,n}_{M,h}\Big]\right)^2\Big|\widetilde{Z'}^{\lambda,n}_{M,h}\Big]\Bigg]\nonumber\\	
		&=\left(\frac{1}{1-u}\right)^2\frac{C}{N}\E\bigg[\E\Big[\big| f(\widetilde{Z'}^{\lambda,n,d-1}_{M,h},S)\big|^2\Big| \widetilde{Z'}^{\lambda,n}_{M,h}\Big]\bigg]\leq \left(\frac{1}{1-u}\right)^2\frac{C}{N}\E\Big[1+|S|^2\Big]\leq \frac{C}{N}f\big(1+C^{'4}_{(u,t)}\big),\label{z'first}
	\end{align}
	with $C^{'4}_{(u,t)}$ given in \eqref{C'4}.
	Here the third to last step follows by Jensen's inequality and tower property. The penultimate step follows by the boundedness of $\nabla_r f(r,S)$ and Lipschitzness of $\nabla_s f(r,S)$. In the last step, we used that $S$ has finite fourth moment, and equation \eqref{eq:4.moment}.
	Finally, putting \eqref{z'}, \eqref{z'second} and \eqref{z'first} together yields the lemma.  
\end{proof}
We now investigate the second term on the right hand side of \eqref{decomp2}.

\begin{lemma}\label{term1l}
	Under the conditions of \Cref{thm:avar}, for all $t,\gamma>0, u\in(0,1), M,N \in \N^\star$ and $\lambda>1$, 
	we have 
	\begin{align*}
\E \bigg[\bigg|\frac{1}{N} \sum_{n=1}^N L(\widetilde{Z}^{\lambda,n}_{M,h})-\frac{1}{N}\sum_{n=1}^N L(\widehat{Z}^{\lambda,n}_t)\bigg|^2\bigg]
\leq 
C^{'6}_{(u)}h^2+C^{'7}_{(u,M,\lambda)}\gamma^2,
\end{align*}
where constants $C^{'6}_{(u)}$, and $C^{'7}_{(u,M,\lambda)}$  are given in Equations \eqref{C'6} and \eqref{C'7}.
\end{lemma}
\begin{proof}
	Following the same argument as in the proof of Lemma 3.3, we have
	\begin{align}
		&\E \bigg[\bigg|\frac{1}{N} \sum_{n=1}^N L(\widetilde{Z}^{\lambda,n}_{M,h})-\frac{1}{N}\sum_{n=1}^N L(\widehat{Z}^{\lambda,n}_t)\bigg|^2\bigg]
		\leq C\bigg(\frac{1}{N}\sum_{n=1}^N \E\Big[\Big(\widetilde{L}(\widetilde{Z}^{\lambda,n}_{M,h})-\widetilde{L}(\widehat{Z}^{\lambda,n}_{t})\Big)^2 \Big]+\frac{\gamma^2}{N}\sum_{n=1}^N \E \Big[\Big( \big\|\widetilde{Z}^{\lambda,n}_{M,h}\big\|^2-\big\|\widehat{Z}^{\lambda,n}_{t}\big\|^2\Big)^2\Big]\bigg).\label{lem32}
	\end{align}
	Let $Z^\lambda(s)$ be a continuous time approximation of the Euler-Maruyama scheme in (\ref{Milstein}). One way to define such an approximation is by setting 
	\begin{equation*}
		Z^\lambda(s) := z- \int_{0}^{n_s}\nabla L(Z^\lambda(r))\diff r +\int_0^{n_s}\sqrt{2\lambda^{-1}}\diff W_r,
	\end{equation*}	
	with $n_s := \max\{\frac{t}{M^2}n:\frac{t}{M^2}n\leq s, n\in \Z\}$ .
	Note that for each $1\leq i \leq N, 1 \leq m \leq M$, and $t>0$, we have $Z^\lambda(s_m)=\widetilde{Z}_{m,h}^{\lambda,i}$. 
	In other words, $Z^\lambda(s)$ coincides with $\widetilde{Z}_{m,h}^{\lambda,i}$ at the time discretization points.
	For the first term in \eqref{lem32}, we have
	\begin{align*}
		\frac{1}{N}\sum_{n=1}^N \E\Big[|\widetilde{L}(\widetilde{Z}^{\lambda,n}_{M,h})-\widetilde{L}(\widehat{Z}^{\lambda,n}_{t}) |^2 \Big]
		&\leq \frac{4}{(1-u)^2N}\sum_{n=1}^N \E \Big[\big\|\widetilde{Z}_{M,\frac{t}{M^2}}^{\lambda,n}-\widehat{Z}_t^{\lambda,n}\big\|^2\Big]\\
		&\leq \frac{4}{(1-u)^2N}\sum_{n=1}^N \E \Big[\sup_{0 \leq s \leq t} \big\|Z^\lambda(s)-\widehat{Z}_s^{\lambda,n}\big\|^2\Big],
	\end{align*}
	where we used that $\widetilde{L}$ is $\frac{2}{1-u}$--Lipschitz in the first inequality. 
	By standard results on the error estimation for SDE approximations, see e.g.\ \citep[Theorems 10.3.5 and 10.6.3]{num}, we have 
	\begin{equation}
		\E \Big[\sup_{0 \leq s \leq t} |Z^\lambda(s)-\widehat{Z}_s^{\lambda,n}|^4\Big]
		\leq Ch^4\label{mil}
	\end{equation}
	for some constant $C>0$.
	Thus, we have 
	\begin{equation}\label{tildeeuler}
		\frac{1}{N}\sum_{n=1}^N \E\Big[|\widetilde{L}(\widetilde{Z}^{\lambda,n}_{M,h})-\widetilde{L}(\widehat{Z}^{\lambda,n}_{t}) |^2\Big] \leq  \frac{Ch^2}{(1-u)^2}.
	\end{equation}

	For the second term on the right hand side of \eqref{lem32}, by Cauchy-Schwartz inequality, we have
	\begin{align}
		\E \Big[\Big( \big\|\widetilde{Z}^{\lambda,n}_{M,h}\big\|^2-\big\|\widehat{Z}^{\lambda,n}_{t}\big\|^2\Big)^2\Big]\bigg)
		& \leq  \E\Big[\big\| \widetilde{Z}^{\lambda,n}_{M,h}\big\|^4+\big\|\widehat{Z}^{\lambda,n}_{t}\big\|^4\Big]^{\frac{1}{2}} \E\Big[ \big\|\widetilde{Z}^{\lambda,n}_{M,h}-\widehat{Z}^{\lambda,n}_{t}\big\|^4\Big]^{\frac{1}{2}}\nonumber\\
		&\leq  \E\Big[\big\| \widetilde{Z}^{\lambda,n}_{M,h}\big\|^4]+ \big\|\widehat{Z}^{\lambda,n}_{t}\big\|^4\Big]^{\frac{1}{2}}\E\Big[ \sup_{0 \leq s \leq t}\big\|Z^\lambda (s)-\widehat{Z}^{\lambda,n}_{s}\big\|^4\Big]^{\frac{1}{2}}\nonumber\\
		&\leq C \E\Big[\big\| \widetilde{Z}^{\lambda,n}_{M,h}\big\|^4 +\big\|\widehat{Z}^{\lambda,n}_{t}\big\|^4\Big]^{\frac{1}{2}}h^2\label{eulersol}
	\end{align}
	where we used equation \eqref{mil} in the last step.

\medskip

	It remains to control the fourth moment of $\widehat{Z}^{\lambda,n}_{t}$, since that of $\widetilde{Z}^{\lambda,n}_{M,h}$ was bounded in Equation \eqref{eq:4.moment}.
	Since $\widehat{Z}^{\lambda,n}_{t}$ solves the SDE 
	$$\diff\widehat{Z}^{\lambda,n}_t=-\nabla L(\widehat{Z}^{\lambda,n}_t)\diff t+\sqrt{2 \lambda^{-1}}\diff \widetilde{W}^n_t,$$ 
	with a linearly growing drift, the following bound on the fourth moment of the solution follows by standard arguments:
	\begin{align}\label{**}
		\E\Big[\big\|\widehat{Z}^{\lambda,n}_t\big\|^4] &\leq C\left(1+ \frac{t}{(1-u)^4}+\frac{1}{\lambda^2}t^2\right)e^{\gamma^4Ct},
	\end{align}
 where $C$ is a constant depending only on the Lipschitz constant of $\nabla L$.
 	We omit the proof.
	Putting equations \eqref{tildeeuler}, \eqref{eulersol}, \eqref{**} and \eqref{eq:4.moment} together, and recalling that $t=hM^2$, we have the result of the lemma.
	\end{proof}
Now we move on to analyzing the third term in \eqref{decomp2}.
\begin{lemma}\label{term2l}
	Under the conditions of \Cref{thm:avar}, for all $h,t>0, u\in(0,1), N \in \N^\star$, $0<\gamma<1$, and $\lambda>1$, we have 
	\begin{align*}
		\E \bigg[\bigg|\frac1N \sum_{n=1}^N L(\widehat{Z}^{\lambda,n}_t)-\E[L(Z_t^\lambda)]\bigg|^2\bigg]
		\leq \frac{1}{N}C^{'8}_{(u,M,\lambda,t)}
	\end{align*}
	for $C^{'8}_{(u,M,\lambda,t)}$ given in Equation \eqref{C'8}.
\end{lemma}
\begin{proof}
	Since $(\widehat{Z}^{\lambda,n})_{n\ge 1}$ are i.i.d. copies of $Z_t^\lambda$, a standard law of large numbers argument gives 
	\begin{align*}
		\E \bigg[\bigg|\frac1N \sum_{n=1}^N L(\widehat{Z}^{\lambda,n}_t)-\E[L(Z_t^\lambda)]\bigg|^2\bigg]
		\leq\frac1N \mathrm{Var}(L(Z^\lambda_t))
	\end{align*}
	where $\mathrm{Var}(L(Z^{\lambda}_t))$ is the variance of $L(Z^\lambda_t)$.
	Since $\nabla L$ is $(\frac{1}{1-u}+\gamma)$--Lipschitz, it follows by \citep[Corollary 5.11]{Poincare}, that the law of $Z^\lambda_t$ satisfies the Poincar\'e inequality.
	That is,
	$$\mathrm{Var}(L(Z_t^\lambda))\leq \frac{2}{\lambda(\frac{1}{1-u}+\gamma)^2}\E \Big[\big\|\nabla L(Z_t^\lambda)\big\|^2\Big].$$
	Now using that $\nabla \widetilde{L}$ is $\frac{2}{1-u}$--Lipschitz, we have $\big\|\nabla L(x)\big\|\leq \|\nabla \widetilde{L}(0)\|+(\frac{2}{1-u}+\gamma)\|x\|$. 
	Thus, it follows that
	\begin{align*}
		\mathrm{Var}(L(Z_t^\lambda))&\leq \frac{2}{\lambda(\frac{1}{1-u}+\gamma)^2}\E\bigg[\bigg|\|\nabla \widetilde{L}(0)\|+\Big(\frac{2}{1-u}+\gamma\Big)\big\|Z^\lambda_t\big\|^2\bigg|\bigg]\\
		&\leq C\bigg(\frac{1}{\lambda(\frac{1}{1-u}+\gamma)^2}+\frac{1}{\lambda}\Big(\E\Big[\big\|Z_t^\lambda\big\|^4\Big]\Big)^{\frac{1}{2}}\bigg)\\
		&\leq C\bigg(\frac{1}{\lambda(1+\gamma)^2}+\frac{1}{\lambda}\bigg(1+\frac{\sqrt{t}}{(1-u)^2}+\frac{t}{\lambda}\bigg) e^{\gamma^4C t}\bigg),
	\end{align*}
	where the last inequality follows from \eqref{**}.
	Taking $\gamma \geq 0$ yields the result of the lemma.
\end{proof}
In the next lemma we analyze the fourth term in \eqref{decomp2}.
Essentially, this concerns the rate of convergence of the law $\mu_t^\lambda$ of the solution of the Langevin equation to its invariant measure $\mu^\lambda_\infty$.
\begin{lemma}\label{term3l}
	For all $t>0, u,\gamma\in(0,1)$, $\lambda>1$ and an initial position $Z^\lambda_0=z\in \R^{d+1}$, we have 
	\begin{align*}
		\bigg|\E L(Z^\lambda_t)-\int L(x)\diff \mu_\infty^\lambda\bigg|^2 &\leq C^{4}_{(u,\lambda)}e^{- tC^{5}_{(\lambda)}}+ \gamma^2 C^{'9}_{(\lambda,t)},
	\end{align*}
 where constants $C^{4}_{(u,\lambda)}, C^{5}_{(\lambda)}$, and $C^{'9}_{(\lambda,t)}$ are given in equations \eqref{C4} - \eqref{C'9}.
\end{lemma}
\begin{proof}
	The investigation of convergence rates to the invariance measure is an active research area, see e.g.\ \citep{Bol-Gen-Gui12}.
	In the present case with non-convex potential functions $\nabla L$, this follows from the so-called coupling by reflection arguments of \citeauthor*{EberleCRAS} \cite{EberleCRAS}.
	In fact, by \Cref{ass.f.S}.\,$(iv)$, it follows from \cite[Corollary 2.1]{EberleCRAS}, (see also \cite[Corollary 2]{EberlePTRF}) that there is a constant $C_{(\lambda)}>0$ depending only on $d,\gamma$ and $\lambda$ such that
	\begin{equation}
	\label{eq:spectral.gap1}
	 	\cW_1(\mu^\lambda_t,\mu_\infty^\lambda)\le C_{(d,\lambda)}e^{-tC_{(d,\lambda)}}\cW_1(\delta_z,\mu^\lambda_\infty) \quad \text{for all } t>0.
	 \end{equation}

\medskip

	It now remains to bound $|\E L(Z^\lambda_t)-\int L(x)d\mu_\infty^\lambda|$ by $\mathcal{W}_1(\mu^\lambda_t, \mu^\lambda_\infty)$.
	Let $\hat\alpha \in \Gamma(\mu_\infty^\lambda,\mu_t^\lambda)$ be an optimal coupling of $\mu_\infty^\lambda$ and $\mu_t^\lambda$, i.e. such that 
	$$\mathbb{E}_{\hat\alpha}\|X-Y\|=\inf_{\alpha \in \Gamma(\mu_\infty^\lambda,\mu_t^\lambda)}\mathbb{E}_\alpha\|X-Y\|,$$
	see e.g.\ \cite{Villani} for existence of $\hat\alpha$.
	Above, we denote by $\E_\alpha[\|X-Y\|]$ the expectation under $\alpha$ of $\|X-Y\|$, with $(X,Y) \sim \alpha$.
	By definition of $L$ and Lipschitz--continuity of $\widetilde L$, it holds
	$$|L(x)-L(y)| 
		\leq \frac{2}{1-u}\|x-y\|+ \frac{\gamma}{2}(\|x\|^2 + \|y\|^2). 
		$$  
	Taking the expectation with respect to $\hat\alpha \in \Gamma(\mu_\infty^\lambda,\mu_t^\lambda)$ we have
	\begin{align}
	\nonumber
		\mathbb{E}_{\hat\alpha} \Big[|L(X)-L(Y)|\Big]
		&\leq \frac{2}{1-u}\mathbb{E}_{\hat\alpha} \Big[\|X - Y\|\Big] + \frac{\gamma}{2} \E_{\hat\alpha}\Big[\|X\|^2+ \|Y\|^2\Big]\\	
		\label{eq:spectral.gap.2}
		&\le \frac{2}{1 - u}\cW_1(\mu_\infty^\lambda,\mu_t^\lambda) + \frac{\gamma}{2}\Big(\E[\|Z^{\lambda}_\infty\|^2] + \E[\big\|Z^\lambda_t\big\|^2]\Big), 
	\end{align}
where $Z^{\lambda}_\infty \sim \mu^{\lambda}_\infty$.
	As in the proof of \Cref{term1l}, see e.g.\ \Cref{**}, $Z^\lambda_t$ has second moment bounded by a constant $C_{(\lambda,t)}>0$.
	Concerning the term $\mathbb{E}[\|Z^\lambda_\infty\|^2]$, note that it holds $$\mathbb{E}[\|Z^\lambda_\infty\|^2]=\int \|x\|^2 \diff\mu_\infty^\lambda=\int \|x\|^2\frac{e^{-\lambda L(x)}}{\int e^{-\lambda L(a)}\diff a}\diff x.$$
	Since $\widetilde{L}$ is $\frac{2}{1-u}$-Lipschitz, $|\widetilde{L}(x)|\leq C+\frac{2}{1-u} \|x\|$ for some constant $C$. 
	We thus have 
	\begin{align}
    	\int \|x\|^2e^{-\lambda L(x)}\diff x = \int \|x\|^2 e^{-\lambda(\widetilde{L}(x)+\gamma \|x\|^2/2)}\diff x \leq \int \|x\|^2e^{C\lambda+2\lambda\|x\|/(1-u)- \|x\|^2/2}\diff x<\infty.\label{eq:momend.mu.infty}
	\end{align}
	Using the same  argument, $0<\int e^{-\lambda L(a)}\diff a<\infty$. Therefore, we have $\mathbb{E} [\|Z^\lambda_\infty\|^2]<\infty$.
	Combining this with \eqref{eq:spectral.gap1} and \eqref{eq:spectral.gap.2} yields the desired result.
\end{proof}
\begin{remark}
	If the function $f$ is convex, it follows that $L$ is strongly convex in the sense that $\nabla^2 L\geq \gamma I_{d}$, where $I_{d}$ is the $(d)\times(d)$ identity matrix.
	In this case, the exponential convergence to equilibrium follows by standard arguments, see e.g.\ \citep{Bol-Gen-Gui12}.
	It fact, we have the following bound is second order Wasserstein distance:
	\begin{equation}
	\label{eq:spectral.gap1222}
		\mathcal{W}_2^2(\mu^{\lambda}_t, \mu^{\lambda}_\infty) \le e^{-2\gamma t}\mathcal{W}_2^2(\delta_z, \mu^\lambda_\infty).
	\end{equation}
	In this case, a slight modification of the above arguments allow to get the bound
	\begin{align*}
		\bigg|\E L(Z^\lambda_t)-\int L(x)\diff \mu_\infty^\lambda\bigg|^2 &\leq C\left(\frac{1}{(1-u)^2}+1\right)e^{-\gamma t}\mathcal{W}_2^2(\delta_z, \mu^\lambda_\infty).
	\end{align*}
\end{remark}
The estimation of the fifth term in the decomposition \eqref{decomp2} is an immediate consequence of the existence of second moment of the invariant measure $\mu_\infty^\lambda$ obtained in the proof of the preceding lemma.
In fact, by definition of $L$ and $\widetilde L$ we have
\begin{equation}
\label{eq:diff.L.Ltilde}
	\bigg|\int L (x) \diff\mu_\infty^\lambda-\int \widetilde{L}(x) \diff\mu_\infty^\lambda \bigg|^2= \left(\frac{\gamma}{2}\right)^2\int \|x\|^2 \diff\mu_\infty^\lambda
\leq C\gamma^2,
\end{equation}
for a constant $C>0$.
We conclude the proof of the theorem with the following lemma estimating the last term on the right hand side of \eqref{decomp2}.

\begin{lemma}\label{term5l}
	Under the conditions of \Cref{thm:avar}, we have
	\begin{align} \label{error5}
		\bigg|\int \widetilde{L}(x) \diff\mu_\infty^\lambda  - \overline{\mathrm{AVaR}(f)}\bigg|^2
		\leq  C\gamma^2+\frac{C^{'9}_{(u)}}{\lambda^2},
	\end{align}
where $C^{6}_{(u)}$ is given in Equation \eqref{C6}.
\end{lemma}
\begin{proof}
	First recall, see \Cref{opt} and \Cref{Hwang} that $\overline{\mathrm{AVaR}}(f) =  \widetilde{L}(z^*)$ where $z^* = (r^*, m^*)$ is the optimizer in \eqref{opt}.
	Now, consider the differential entropy 
	\begin{align*}
	-\int \mu_{\infty}^{\lambda}\log \mu_{\infty}^{\lambda} \diff x
		&=-\int \frac{e^{-\lambda L(x)}}{\int e^{-\lambda L(u)}\diff u} \log \frac{e^{-\lambda L(x)}}{\int e^{-\lambda L(u)}\diff u}\diff x\\
		&=-\int \frac{e^{-\lambda L(x)}}{\int e^{-\lambda L(u)}\diff u}\left(-\lambda L(x)-\log \int e^{-\lambda L(u)}\diff u\right)\diff x\\
		&=\lambda \int L(x) \diff \mu_{\infty}^{\lambda}+\log \int e^{-\lambda L(x)}\diff x\\
		&=\lambda \int \left(\widetilde{L}(x)+\frac{\gamma}{2}\|x\|^2\right)\diff \mu_\infty^\lambda +\log \int e^{-\lambda \widetilde{L}(x)-\frac{\lambda\gamma}{2}\|x\|^2}\diff x.
	\end{align*}
	Rearranging the terms gives the following expression for the integral of $\widetilde L$:
	\begin{equation}\label{laplace}
		\int \widetilde{L}(x) \diff\mu_{\infty}^{\lambda}
		=-\frac{\gamma}{2}\int \|x\|^2 \diff \mu_\infty^\lambda-\frac{1}{\lambda}\int \mu_{\infty}^{\lambda}\log \mu_{\infty}^{\lambda}\diff x-\frac{1}{\lambda}\log \int e^{-\lambda \widetilde{L}(x)-\frac{\lambda\gamma}{2}\|x\|^2}\diff x.
	\end{equation}
	Since for any continuous random variable, a Gaussian distribution with the same second moment maximizes the differential entropy (\citep[Theorem 10.48]{Yeung}), it holds
	\begin{equation}\label{ineq}
		- \int \mu_{\infty}^\lambda \log \mu_\infty^\lambda \diff x \leq \frac{1}{2}\log\left( ({2 \pi e})^{d+1} \int \|x\|^2\diff \mu_\infty^\lambda\right),
	\end{equation}
	and using that $\int \|x\|^2\diff \mu_\infty^\lambda<\infty$  (see equation \eqref{eq:momend.mu.infty}) and subtracting $\widetilde{L}(z^*)$ from both sides of \eqref{laplace}, we have
	\begin{align*}
		\int \widetilde{L}(x) \diff \mu_\infty^\lambda - \widetilde{L}(z^*)
		&\leq  C\left(-\gamma +\frac{1}{\lambda} -\frac{1}{\lambda}\log \int e^{-\lambda \widetilde{L}(x)-\frac{\gamma}{2}\|x\|^2} \diff x-\widetilde{L}(z^*)\right)\\
		&\leq C\left(\gamma +\frac{1}{\lambda}-\frac{1}{\lambda}\log \left(e^{-\lambda \widetilde{L}(z^*)}\int e^{-\lambda (\widetilde{L}(x)-\widetilde{L}(z^*)-\frac{\gamma \lambda}{2}\|x\|^2}\diff x\right)-\widetilde{L}(z^*)\right)\\
		&= C\left(\gamma +\frac{1}{\lambda}-\frac{1}{\lambda}\log \int e^{-\lambda(\widetilde{L}(x)-\widetilde{L}(z^*))-\frac{\lambda\gamma}{2}\|x\|^2}\diff x\right).
	\end{align*}
	Hence, using the fact that $\nabla \widetilde{L}$ is $\frac{1}{1-u}$-Lipschitz, we can estimate the exponent in the integral above as 
	$$|\widetilde{L}(x)-\widetilde{L}(x^*)-\nabla \widetilde{L}(z^*)\cdot (x-z^*)|\leq \frac{1}{2(1-u)}\|x-z^*\|^2.$$ 
	Thanks to the above inequality, using $\nabla \widetilde{L}(z^*)=0$ we obtain 
	\begin{align*}
		\int \widetilde{L}(x) \diff\mu_\infty^\lambda - L(z^*)
		&\leq C\bigg(\gamma +\frac{1}{\lambda}-\frac{1}{\lambda} \log \int e^{-\frac{ \lambda}{2(1-u)}\|x-z^*\|^2-\frac{\gamma \lambda}{2}\|x\|^2}\diff x\bigg)\\
		&=C\bigg(\gamma +\frac{1}{\lambda}-\frac{1}{\lambda} \log\bigg(\sqrt{\frac{2\pi}{\lambda(\gamma+\frac{1}{(1-u)}) }}e^{-\frac{\gamma \lambda \|z^*\|^2}{2(1-u)(1/(1-u)+\gamma)}}\bigg)\bigg)\\
		&\leq C\bigg(\gamma +\frac{1}{\lambda}-\frac{1}{\lambda}\log\bigg(\sqrt{\frac{2\pi}{\lambda(\gamma+\frac{1}{1-u}) }}\bigg)+\frac{\gamma}{(1-u)(\frac{1}{1-u}+\gamma)}\bigg).
	\end{align*}
	For the other side of the inequality, since $z^*$ is a minimizer of $\widetilde{L}$, we have \begin{align*}
		\widetilde{L}(x^*)-\int \widetilde{L}(x)\diff \mu_\infty^\lambda
		&\leq \widetilde{L}(z^*)-\widetilde{L}(z^*) \int \diff\mu_\infty^\lambda  = 0.
	\end{align*}
	Using $0<\gamma<1$  and $\lambda \geq 1$ concludes the proof.
\end{proof}

\section{Rate for general law invariant convex risk measures}
\label{sec:proofs.genRM}
We now focus on the estimation of the approximation error of general convex risk measures.
As explained in Section \ref{eq:approx.meth}, the main argument for the derivation of the rate is the representation of the (law invariant) convex risk measure with respect to $\mathrm{AVaR}$.
One technical difficulty is that this representation involves an integral with respect to the risk aversion level $u$ of $\mathrm{AVaR}$.
Notice that both the functions $L$ and $\widetilde L$ as well as the Markov chain $\widetilde Z^{\lambda, n}_{m,h}$ in the approximation scheme depend on $u$.
We will make this dependence explicit in this section by writing
$L^u$, $\widehat{Z}^{\lambda,n,u}_{t}$, and $\widetilde{Z}^{\lambda,n,u}_{M,h}$ for the function and processes defined in \eqref{eq:object.L}, \eqref{eq:def.Z.lambda} and \eqref{Milstein} respectively.

\subsection{Proof of Theorem \ref{thm:general.rm}}
This subsection covers the proof of Theorem \ref{thm:general.rm}.
Recall that we approximate a general law invariant convex risk measure by 
\begin{equation*}
	\widetilde{\rho}^{\delta}(f):=\esssup_{\gamma \in \mathcal{M}: \beta(\gamma)\leq b} \bigg(\int_0^\delta \widetilde{\text{AVaR}}_u(f)\gamma(\diff u)-\beta(\gamma)\bigg),
\end{equation*}
with $\widetilde{\text{AVaR}}_u(f)=\frac{1}{N}\sum_{n=1}^N \overline{\ell^u}(\widetilde{Z}_{M,h}^{\prime,\lambda,n,u})$. 
Further define 
\begin{equation}\label{rhodelta}
	\rho^{\delta}(f) := \inf_{r\in A}\sup_{\gamma\in \mathcal{M}: \beta(\gamma)\leq b}\bigg(\int_0^\delta \text{AVaR}_u(f(r, S))\gamma(\diff u)-\beta(\gamma)\bigg), \quad \delta \in (0,1).
\end{equation}
To begin, we decompose the approximation error into two parts:
\begin{align}\label{approxerror}
	\E\Big[|\rho(f)-\widetilde{\rho}^{\delta}(f)|^2\Big] \leq 2|\rho(f)-\rho^\delta(f)|^2 + 2 \E\Big[|\rho^\delta(f)- \widetilde{\rho}^{\delta}(f)|^2\Big].
\end{align}
Let us estimate the first term. First, we will show  that for all $\delta\in (0,1)$, $\rho^\delta$ satisfies 
\begin{align}
\label{eq:rep.rho.delta}
	\rho^{\delta}(f)&=\sup_{\gamma\in \mathcal{M}: \beta(\gamma)\leq b}\bigg(\int_0^\delta \inf_{r\in A}\text{AVaR}_u(f(r,S))\gamma(\diff u)-\beta(\gamma)\bigg) \\\notag
	&=\sup_{\gamma\in \mathcal{M}: \beta(\gamma)\leq b}\bigg(\int_0^\delta \overline{\text{AVaR}}_u(f)\gamma(\diff u)-\beta(\gamma)\bigg).
\end{align}
Note that going from the definition in \eqref{rhodelta} to the expression in \eqref{eq:rep.rho.delta} requires interchaging the infimum and the supremum and then the infimum and the integral in the definition given in \eqref{rhodelta}. Since the supremum is taken over a compact set and the function $(r,\gamma)\mapsto \int_{0}^\delta\text{AVaR}_u(f(r,S))\gamma(du)-\beta(\gamma)$ is continuous in $r$ and upper semi--continuous and concave in $\gamma$, by \citeauthor*{kyfan}'s minimax theorem, see \cite[Theorem 2]{kyfan}, we can interchange the supremum and the infimum. To interchange the infimum and the integral, we apply Rockafellar's interchange theorem \citep{Rockafellar}.
Thus, using \Cref{bounded} and \Cref{eq:rep.rho.delta}, it holds that
\begin{equation*}
	|\rho(f)-\rho^\delta(f)|^2 \leq \sup_{\gamma\in \mathcal{M}: \beta(\gamma)\leq b}\int_\delta^1 |\overline{\text{AVaR}}(f)|^2\gamma(\diff u).
\end{equation*}
Let us partition the interval $[\delta,1)$ into $\bigcup_{n\geq1} I_n$, where $I_n= \big[ \delta+(1-2^{-n+1})(1-\delta),\delta+(1-2^{-n})(1-\delta) \big)$ for every $n \geq 1$. 
Defining $$\Gamma_b(I_n):= \sup_{\gamma \in \mathcal{M} \text{ s.t } \beta(\gamma)\leq b}\gamma(I_n)$$,
we obtain the estimation
\begin{align}
	|\bar\rho(f)-\rho^\delta(f)|^2
&\leq \sum_{n \geq 1}\Gamma_b(I_n)\sup_{u \in I_n}|\overline{\text{AVaR}}_u(f)|^2.\label{rhobardelta}
\end{align}
Now by \citep[Lemma 4.5 and Lemma 4.3]{Ludo}, we have 
\begin{equation}\label{gammain}
	\Gamma_b(I_n)\leq C ((1-\delta)2^{-n})^{1/q}\quad \text{and}\quad 	\left|\overline{\text{AVaR}}_u(f)\right|\leq \frac{C}{(1-u)^{1/p}},
\end{equation} 
for every $p \in (1,\infty)$.
Therefore, for any given $p \in (1,\infty)$ and $n \in  \mathbb{N}$, it holds 
\begin{equation}
	\sup_{u \in I_n} \left|\overline{\mathrm{AVaR}}_u(f)\right|^2 \leq \frac{C}{((1-\delta)2^{-n})^{2/p}}.\label{supavar}
\end{equation}
Choosing $p=4q$ and using \eqref{rhobardelta}, \eqref{gammain} and \eqref{supavar}, we have
\begin{align*}
	|\bar\rho(f)-\rho^\delta(f)|^2
	&\leq \sum_{n \geq 1}((1-\delta)2^{-n})^{1/q} \frac{C}{((1-\delta)2^{-n})^{2/p}}\\
	&\leq C (1-\delta)^{1/q-2/p}\sum_{n\geq 1}2^{n(2/p-1/q)}\\
	&\leq C(1-\delta)^{1/2q}\sum_{n \geq 1} 2^{-n/2q}.
\end{align*}
Since $q \in (1,\infty)$, it holds $\sum_{n \geq 1} 2^{-n/2q}\leq C$ for some universal constant $C$, and thus
\begin{equation}
\label{eq:bound.rho.delta}
	|\bar\rho(f)-\rho^\delta(f)|^2\leq C(1-\delta)^{1/2q}.  
\end{equation}

\medskip

For the second term in equation \eqref{approxerror}, we use the error rate for the estimation of AVaR. Let $\mathcal{M}^R$ be the set of all random probability measures on $[0,1)$. We have 
\begin{align*}
\E|\rho^\delta(f)- \widetilde{\rho}^{\delta}(f)|^2\Big] 
&\leq \E \Big[\esssup_{\gamma \in \mathcal{M}} \int_0^\delta |\widetilde{\text{AVaR}}_u-\text{AVaR}_u|^2\gamma (du)\Big]\\
&\leq \E \Big[\esssup_{\gamma \in \mathcal{M}^R} \int_0^\delta |\widetilde{\text{AVaR}}_u-\text{AVaR}_u|^2\gamma (du)\Big].
\end{align*}
For each random measure $\gamma^i \in \mathcal{M}^R$, define the corresponding random variable $A^i:=\int_0^\delta |\widetilde{\text{AVaR}}_u-\text{AVaR}_u|\gamma^i(du)$, and let $\mathcal{A}$ be the set of all such random variables. To make use of the error rate for the estimation of AVaR, we first show the set of random variables $\mathcal{A}$ is directed upward, i.e. for any pair of random variables $A^i, A^j \in \mathcal{A}$, there exists $\widetilde{A} \in \mathcal{A}$ with $\widetilde{A}\geq \max\{A^i, A^j\}$. Then, by the following theorem, we can rewrite $\esssup \mathcal{A}=\lim_{n}A^n$ for some increasing sequence $(A)_n \in \mathcal{A}$. 
\begin{theorem}\citep{Fol-Sch}\label{upward}
If $\mathcal{A}$ is directed upward, there exists an increasing sequence $A^1\leq A^2 \leq \cdots \in \mathcal{A}$ such that $\esssup \mathcal{A}=\lim_{n}A^n$ $\P$-almost surely.
\end{theorem}
To show $\mathcal{A}$ is directed upward, for any $\gamma^i, \gamma^j \in {\mathcal{A}},$ define the set of events 
\begin{equation}
B:=\Bigg\{\omega : \int_0^\delta |\widetilde{\text{AVaR}}_u-\text{AVaR}_u|\gamma^i(\omega,du) \geq \int_0^\delta |\widetilde{\text{AVaR}}_u-\text{AVaR}_u|\gamma^j(\omega,du)\Bigg\}, 
\end{equation}
and the random measure
\begin{equation}
\widetilde{\gamma}(\omega, du):=\gamma^1(\omega, du) \mathds{1}_{\omega \in B} +\gamma^2(\omega, du) \mathds{1}_{\omega \in B^C}.    
\end{equation}
Then, we have $\widetilde{A}:=\int |\widetilde{\text{AVaR}}_u-\text{AVaR}_u|\widetilde{\gamma}(\omega,du) \geq \max\{A^i, A^j\}$ and $\widetilde{A} \in \mathcal{A}$. Therefore, $\mathcal{A}$ is directed upward. By Theorem \ref{upward}, there exists an increasing sequence $(A^n)_n$  in $\mathcal{A}$ with $\esssup \mathcal{A}=\lim_{n}A^n$ $\P$-almost surely, and we have  
\begin{align*}
\E\Big[|\rho^\delta(f)- \widetilde{\rho}^{\delta}(f)|^2\Big]
&=\E \Big[
\lim_{n}\int_0^\delta |\widetilde{\text{AVaR}}_u-\text{AVaR}_u|^2\gamma^n(du)\Big]\\
&=\lim_{n}  \int_0^\delta \E\Big[|\widetilde{\text{AVaR}}_u-\text{AVaR}_u|^2\Big]\gamma^n(du)\\
&\leq \sup_{n \in \N} \int_0^\delta \E\Big[|\widetilde{\text{AVaR}}_u-\text{AVaR}_u|^2\Big]\gamma^n(du)\\
&\leq \sup_{u \in [0,\delta]}\E\Big[|\widetilde{\text{AVaR}}_u-\text{AVaR}_u|^2\Big]\sup_{n \in \N}\int_0^\delta \gamma^n(du),
\end{align*}
where we used monotone convergence theorem and Fubini's theorem in the second line. 
Using Theorem 2.3 and taking the supremum over $u \in (0,\delta)$, we have the result of the theorem. 

\section{Appendix}
Here we give explicit formulas for constants in the proof of Theorem 2.3.
\begin{align}
C^{'1}_{(u,t,\lambda)}
&=C\bigg(\Big(\frac{1}{1-u}\Big)^2+1\bigg)\frac{1}{t}\bigg(2^{1/t}+\Big(\frac{1}{(1-u)^4}+\frac{C_d}{\lambda^2}\Big)\cdot 2^{1/t}+1\bigg)\exp\bigg(C\frac{1}{t}\Big(\frac{1}{(1-u)^2}+1\Big)\bigg)\label{C'1}\\
C^{'2}_{(u,t,\lambda)}
&=C\bigg(3^{1/t}+\Big(\frac{1}{(1-u)^4}+1+\frac{C_d}{\lambda^2}\Big)\cdot 3^{1/t}+1\bigg)\label{C'2}\\
C^{'3}_{(u,t)}
&=C\frac{1}{(1-u)^4}\frac1t\exp\bigg(\frac{C}{t}\Big(\frac{1}{1-u}+1\Big)^4\bigg)\label{C'3}\\
C^{'4}_{(u,t)}
&=C\frac{1}{(1-u)^2}\bigg(2^{1/t}+\Big(\frac{1}{(1-u)^4}+\frac{1}{\lambda^2}C_d\Big)\cdot 2^{1/t}\bigg)\label{C'4}\\
C^{'5}_{(u,t,\lambda)}
&=\bigg(2^{1/t}+\Big(\frac{1}{(1-u)^4}+\frac{1}{\lambda^2}C_d\Big)\cdot 2^{1/t}\bigg)\label{C'5}\\
C^{'6}_{(u)}
&=C\frac{1}{(1-u)^4}\label{C'6}\\
C^{'7}_{(u,t,\lambda)}&=C\bigg(2^{1/t}+\Big(\frac{1}{(1-u)^4}+\frac{1}{\lambda^2}C_d\Big)\cdot 2^{1/t}+\Big(1+\frac{1}{t(1-u)^4}+\frac{1}{t\lambda^2}e^{1/t}\Big)^{1/2}\bigg)\label{C'7}\\
C^{'8}_{(u,t,\lambda)}
&=C\bigg(\frac{1}{\lambda(\frac{1}{1-u})^2}+\frac{1}{\lambda(\frac{1}{1-u})^2}\Big(1+\frac{1}{t(1-u)^4}+\frac{1}{t^2\lambda^2}\Big)^{1/2}e^{Ct}\bigg)\label{C'8}\\
C^{4}_{(u,\lambda)}&=\frac{C_{(\lambda)}}{(1-u)^2}\cW^2_1(\delta_z,\mu^\lambda_\infty)\label{C4}\\
C^{5}_{(\lambda)}&=2C_{(\lambda)}\label{C5}\\
C^{'9}_{(u,t,\lambda)}&=(C^{'10})^2+C\bigg(1+\frac{1}{t(1-u)^4}+\frac{1}{t^2\lambda^2}\bigg)e^{C/t}\label{C'9}\\
C^{'10}_{(u,\lambda)}&=\int \|x\|^2e^{C\lambda+2\lambda\|x\|/(1-u)- \|x\|^2/2}\diff x\label{C'10}\\
C^{6}_{(u)}&=C\Big(\log \sqrt{2\pi(1-u)}+1\Big)^2\label{C6}\\
C^1_{(u,t,\lambda)}&=1+C^{'4}_{(u,t)}+C^{'8}_{(u,t,\lambda)}\label{C1}\\
C^2_{(u,t,\lambda)}&=C^{'2}_{(u,t,\lambda)}+C^{'5}_{(u,t,\lambda)}+C^{'7}_{(u,t,\lambda)}+C^{'11}_{(u,t,\lambda)}+C\label{C2}\\
C^3_{(u,t,\lambda)}&=C^{'1}_{(u,t,\lambda)}+C^{'3}_{(u,t)}+C^{'6}_{(u)}\label{C3}\\
C^{7}_{(\delta,t,\lambda)}&=C\frac{1}{(1-\delta)^2}\bigg(2^{1/t}+\Big(\frac{1}{(1-\delta)^4}+\frac{1}{\lambda^2}C_d\Big)\cdot 2^{1/t}\bigg)+C\bigg(\frac{1}{\lambda}+\frac{1}{\lambda}\Big(1+\frac{1}{t(1-\delta)^4}+\frac{1}{t^2\lambda^2}\Big)^{1/2}e^{Ct}\bigg)\label{C7}.
\end{align}

\bibliography{bib_empirical_risk,references,bibliography,references3,bib_empirical_risk2,langevin_paper4}{}

\end{document}